%% file: arxiv_version.tex
\documentclass{article} 
\usepackage{iclr2027_conference,times}

\input{math_commands.tex}

\usepackage{hyperref}
\usepackage{url}
\usepackage[utf8]{inputenc} 
\usepackage[T1]{fontenc}    
\usepackage{booktabs}       
\usepackage{amsfonts}       
\usepackage{nicefrac}       
\usepackage{microtype}      
\usepackage{xcolor}         
\usepackage{amsthm}
\usepackage{thmtools}
\usepackage{algorithm}
\usepackage{algpseudocode}
\usepackage{comment}
\usepackage{enumitem}
\usepackage{amsmath,amssymb}

\usepackage{graphicx}
\usepackage{subcaption}

\usepackage{tikz}               
\usetikzlibrary{calc,arrows.meta,decorations.pathreplacing}
\usetikzlibrary{shapes,positioning,fit}
\usepackage{pgfplots}
\pgfplotsset{compat=1.18}

\newtheorem{theorem}{Theorem}
\newtheorem{lemma}[theorem]{Lemma}

\newtheorem{definition}{Definition}

\newcommand{\M}{\mathcal{M}}

\newcommand{\ZZ}{\mathbb{Z}}

\DeclareMathOperator{\dist}{dist}

\DeclareMathOperator{\fracpart}{frac}
\newcommand{\dTV}{d_{\mathrm{TV}}}

\newcommand{\EE}{\mathbb{E}}
\newcommand{\PP}{\Pr}

\title{Dithered Gaussian Mechanism for Randomness-Efficient Differential Privacy}

\author{%
  Nikita P. Kalinin \\
  Institute of Science and Technology Austria\\
  \texttt{nikita.kalinin@ist.ac.at} \\
   \And
   Rasmus Pagh \\
   BARC \\
  University of Copenhagen \\
   \texttt{pagh@di.ku.dk} \\
  }

\iclrfinalcopy 
\begin{document}

\maketitle

\begin{abstract}
We present the \emph{dithered Gaussian mechanism}, an alternative to the discrete Gaussian mechanism for differential privacy that discretizes the private output rather than the noise distribution itself.
By interpreting this discretization as post-processing of the Gaussian mechanism, our construction directly inherits the privacy guarantees of the standard Gaussian mechanism while avoiding vulnerabilities caused by finite-precision floating-point outputs.
In addition, the mechanism is provably randomness-efficient: by sampling the discretized output values directly, the number of high-quality random bits required for privacy can be reduced significantly and made independent of the noise level.
This is achieved by separating the randomness into two sources: a high-quality source used for the privacy-critical sampling step, and a high-performance public source, possibly known to the adversary, that supplies the additional randomness needed for randomized discretization.
This separation enables the use of cryptographically secure randomness without substantial performance loss.
As an application, we study model training with DP-SGD and show that cryptographically secure noise generation with reduced exposure to floating-point vulnerabilities can be achieved with modest practical overhead.
\end{abstract}

\section{Introduction}

Modern machine learning increasingly relies on sensitive individual-level data, making formal privacy guarantees essential.
Differential privacy provides one such guarantee by limiting how much any single data point can influence the output of an algorithm.
The Gaussian mechanism~\citep*{dwork2006our} is one of the most widely used primitives for ensuring differential privacy.
It is especially central in private machine learning, where it forms the noise-addition step of differentially private stochastic gradient descent (DP-SGD) due to~\citet*{abadi2016deep}.

Despite its ubiquity, the Gaussian mechanism has important practical limitations.
First, it is an idealized continuous mechanism, whereas real implementations run on finite-precision computers.
As a result, a faithful implementation of a continuous Gaussian distribution is impossible: naive samplers can leave detectable ``holes'' in the set of representable floating-point values, creating privacy vulnerabilities of the kind first observed for the Laplace mechanism~\citep{mironov2012significance} and later identified for the Gaussian mechanism itself~\citep{jin2022we}.
A second limitation is the amount of randomness required to sample Gaussian noise.
In the ideal continuous model, exact sampling requires infinitely many random bits, while high-precision implementations can still require many random bits in practice.
This cost becomes substantial at scale: recent efforts to train language models entirely under differential privacy, such as the \emph{VaultGemma} 1B model, required quadrillions of Gaussian random draws and petabytes of randomness to obtain a formal privacy guarantee~\citep{vaultgemma2025}.
At this scale, randomness generation itself can become a bottleneck~\citep{egan2024highspeed}.
Several discrete-valued mechanisms, most prominently the discrete Gaussian mechanism of~\citet*{canonne2020discretegaussian}, address finite-precision concerns by adding discrete noise to discretized vectors.
However, such mechanisms still requires many privacy-critical random bits.
Moreover, because the privacy guarantees do not directly follow from those of the continuous Gaussian mechanism, they require separate privacy analyses for composition and amplification by subsampling, making them difficult to use as drop-in replacements for Gaussian noise in DP-SGD.

\paragraph{Our contributions.}
We propose a drop-in replacement for the Gaussian mechanism that is cryptographically secure, resistant to floating-point attacks, and time-efficient. We achieve these properties through a new approach to noise addition that we call the \emph{dithered Gaussian mechanism}.
%
%
Conceptually, the mechanism can be viewed as adding Gaussian noise to the sensitive vector and then rounding the result to a discrete grid shifted by a public random offset, or \emph{dither}. Conditioned on the dither, the resulting discrete output distribution can be written down explicitly, allowing us to sample from it directly rather than sampling the intermediate continuous Gaussian noise.
The key observation is that the rounded output is a post-processing of the standard Gaussian mechanism, so it directly inherits the privacy guarantees of the Gaussian mechanism.
%
%
In particular, we show that the number of private random bits can be made independent of the Gaussian noise scale, without significantly changing the distribution of the perturbation relative to the Gaussian mechanism.
This substantially reduces the cost of using cryptographically secure randomness.
Since the released output is discrete by construction (conditioned on the public dither), the mechanism avoids finite-precision issues associated with floating-point outputs.
Finally, we apply the mechanism to DP-SGD and show that cryptographically secure noise generation can be achieved with a  modest runtime overhead of about $3\%-7\%$ in our experiments.

\subsection{Related work}
\paragraph{Floating-point vulnerabilities and discrete mechanisms.}
A subtle but important issue in the implementation of differentially private mechanisms is the use of finite-precision floating-point arithmetic. \citet{mironov2012significance} showed that naive implementations of the Laplace mechanism can leave detectable ``holes'' in the set of representable noisy outputs, which can be sufficient to reconstruct a significant part of the underlying data.
Similar finite-precision issues have since been identified for other mechanisms, including the exponential mechanism~\citep{ilvento2020implementing}, Noisy Max~\citep{ding2025avoiding}, and the Gaussian mechanism~\citep{jin2022we}.
A related class of attacks exploits precision-dependent leakage: the precision of the released noisy value may depend on the exponent of the original value, thereby leaking information about the input~\citep{haney2022precision}.
Several implementation-level mitigations have been proposed~\citep{holohan2021secure,haney2022precision} and adopted in practical differential privacy libraries, including diffprivlib~\citep{diffprivlib}, OpenDP~\citep{opendp}, and the ``secure mode'' of Opacus~\citep{yousefpour2021opacus}.
These approaches improve the robustness of floating-point implementations, but they do not remove the underlying mismatch between continuous idealized mechanisms and finite-precision computation.

A more direct way to avoid this kind of vulnerability is to work only with values from a discrete domain.
Several discrete-valued mechanisms have been proposed for this purpose, 
most prominently the discrete Gaussian mechanism~\citep{canonne2020discretegaussian}, along with other similar discrete noise mechanisms~\citep{ghosh2009universally,dwork2006our,agarwal2018cpsgd,agarwal2021skellam,bao2022skellam,chen2022poisson}.
By adding discrete noise and releasing discrete outputs, these mechanisms avoid the floating point artifacts described above.
However, they have two important limitations from the perspective of our work.
First, they can still require substantial amounts of randomness for noise generation, unless one uses coarse approximations.
Second, mechanisms such as the discrete Gaussian approximate Gaussian noise but do not automatically inherit the privacy guarantees of the continuous Gaussian mechanism.
Consequently, analyses for composition, amplification by subsampling, or other accounting procedures generally have to be established separately.
\paragraph{Randomness complexity.}
From a theoretical standpoint, the large randomness requirements of practical differentially private algorithms raise a natural question at the intersection of privacy and complexity: how much randomness is actually necessary for privacy? \citet*{canonne2025randomness} initiated the study of the \emph{randomness complexity of differential privacy}, showing that highly accurate private mechanisms can sometimes be implemented with surprisingly few random bits, in some settings as few as logarithmically many in the number of released statistics.
This contrasts with the common intuition that the randomness required for privacy should scale with the dimension of the output.
Their construction is based on a connection to \emph{secluded partitions} from computational geometry~\citep{vanderwoude2024secluded}, but is not computationally efficient.
More recently, \citet{ghentiyala2026} gave a computationally efficient alternative, at the cost of a polylogarithmic loss in the randomness-error trade-off.

\paragraph{DP-SGD and discrete mechanisms.}
One of the most important applications of Gaussian noise in differential privacy is the training of machine learning models, most commonly through differentially private stochastic gradient descent (DP-SGD)~\citep{song2013stochastic,bassily2014private,abadi2016deep}.
The privacy guarantees of DP-SGD are typically established using privacy accountants that exploit the structure of the sampling procedure, including Poisson subsampling~\citep{abadi2016deep,mironov2017renyi,koskela2020computing,zhu2022optimal,lebeda2025avoiding}, shuffling~\citep{erlingsson2019amplification,feldman2022hiding,feldman2023stronger,wang2023privacy}, sampling without replacement~\citep{balle2018privacy,schuchardt2024unified}, and random allocations such as Balls-in-Bins~\citep{chua2024balls,feldman2025privacy}.
Many of these analyses are developed specifically for the Gaussian mechanism, making it difficult to replace Gaussian noise with a different noise distribution without redoing the privacy analysis.
Discrete mechanisms have also been studied in combination with DP-SGD, especially in federated learning with secure aggregation~\citep{agarwal2018cpsgd,kairouz2021distributed,agarwal2021skellam,bao2022skellam,chen2022poisson}.
In this setting, the main goals are to obtain integer-valued noisy gradient sums, which integrate naturally with secure aggregation, and to reduce communication complexity.

To the best of our knowledge \emph{tight analyses} of subsampled discrete mechanisms are \emph{not available} even for the standard Poisson subsampling scheme.
Existing approaches therefore typically start from concentrated DP or R\'enyi DP guarantees for the underlying discrete mechanism, combine them with a generic subsampling amplification lemma, and then apply generic composition within the chosen privacy framework, as in~\citep{kairouz2021distributed,agarwal2021skellam}.
This contrasts with our approach, where the discretized mechanism inherits Gaussian privacy guarantees by post-processing and can therefore reuse Gaussian-based analyses directly.

\section{Background}\label{sec:background}

In this section, we first introduce the notions of sensitivity, randomness complexity, and privacy that will be used throughout the paper.
Second, we present the properties of the discrete Gaussian mechanism with respect to these notions.

{\bf Sensitivity.}
Consider datasets $X$ that are finite subsets of some set~\(\mathbb{X}\).
We work in the add/remove adjacency notion: two datasets \(X\sim X'\) are neighbors if one can be obtained from the other by inserting or deleting a single data point.
For a dataset of \(n\) elements we write its elements as \(X=\{x^{(1)},\ldots,x^{(n)}\}\subseteq\mathbb{X}\).
For \(p\geq 1\), the global \(\ell_p\)-sensitivity of a function \(f:\mathbb{X}^*\rightarrow\mathbb{R}^d\) is defined as
\[
\Delta_p(f)\ :=\ \sup_{X\sim X'} \|f(X)-f(X')\|_p.
\]

{\bf Randomness complexity.}
The \emph{randomness complexity} of a mechanism \(\M\) is the expected number of random bits required to produce an output of the mechanism.
As noted by~\citet{canonne2025randomness}, this quantity is tightly linked, up to an additive constant, to the Shannon entropy of the mechanism.
When running a sequence of mechanisms independently (keeping public randomness fixed), the total private randomness complexity can be bounded by the sum of entropies plus an additive constant.

{\bf Differential privacy.}
We can release an estimate of $f(X)$ while satisfying \emph{differential privacy} by adding noise to $f(X)$ scaled according to the sensitivity.
Intuitively, differential privacy means that the information released has roughly the same distribution for neighboring datasets $X$ and $X'$.
For example, the Laplace mechanism (for \(p=1\)) attains \((\varepsilon,0)\)-differential privacy and the Gaussian mechanism (for \(p=2\)) attains \((\varepsilon,\delta)\)-differential privacy.
We refer to Appendix~\ref{app:dp-basics} for definitions of differential privacy and properties of basic noise addition mechanisms.

{\bf Discrete Gaussian mechanism.}
Traditionally, differential privacy has been analyzed under the idealized assumption that mechanisms have access to samples from continuous distributions.
Under this viewpoint, the amount of randomness required by a mechanism is typically not accounted for explicitly, and no finite upper bound on the number of random bits is established.
In practice, however, implementations must operate on finite-precision computers and therefore rely on discrete distributions or finite procedures for sampling.
This issue has been studied, for example, by~\citet{vadhan2018finitecomputers} and by~\citet{canonne2020discretegaussian}, who show that discrete analogues of standard private mechanisms admit finite bounds on the required randomness.
A canonical example is the \emph{discrete Gaussian mechanism}~\citep{canonne2020discretegaussian}, which replaces continuous Gaussian noise with noise drawn from a discrete distribution over the integers with probabilities derived from the density of a Gaussian.
Besides being implementable on finite computers, this mechanism provides a natural starting point for studying the randomness complexity of differentially private algorithms.
\begin{definition}[Discrete Gaussian mechanism]
\label{def:discrete_gaussian}
Let $f : \mathcal{X}^n \to \mathbb{Z}^d$ be a function with sensitivity $\Delta_2(f)$.
For $\sigma_{} > 0$, the
discrete Gaussian mechanism with scale $\sigma \Delta_2(f)$ releases
\[
M(x) = f(x) + Z,
\]
where $Z = (Z_1,\dots,Z_d) \in \mathbb{Z}^d$ has independent coordinates, and
each $Z_j$ is sampled from the one-dimensional discrete Gaussian distribution
centered at $0$ with scale $\sigma \Delta_2(f)$, that is,
\[
\mathbb{P}(Z_j = z)
=
\frac{\exp\!\left(-\frac{z^2}{2\sigma^2 \Delta_2(f)^2}\right)}
{\sum_{u \in \mathbb{Z}} \exp\!\left(-\frac{u^2}{2\sigma^2 \Delta_2(f)^2}\right)}
\qquad \text{for all } z \in \mathbb{Z}.
\]
\end{definition}

{\bf Discretization.}
Note that the discrete Gaussian mechanism is defined only for integer-valued functions.
To apply it to general real-valued functions, one must first discretize the values of $f$.
Let $f(x) \in \mathbb{R}^d$ have sensitivity $\Delta_2(f)$, and let $h > 0$ be a discretization parameter.
If each coordinate is rounded to the nearest integer multiple of~$h$, then the sensitivity of the discretized function \(\tilde f\) satisfies
\begin{equation}
\label{eq:sensitivity_rounding}
    \Delta_2(\tilde f) \le \Delta_2(f) + \frac{h\sqrt{d}}{2}.
\end{equation}
However, such deterministic rounding introduces bias.
\emph{Randomizing} the decision whether to round up or down removes this bias, but increases the sensitivity to $\Delta_2(f) + h\sqrt{d}$.
A better trade-off can be obtained by allowing a resampling probability $\beta \in (0,1)$, resulting in a small bias, as described by~\citet{kairouz2021distributed}.
They use a random rotation technique followed by randomized rounding to ensure that, with high probability, the sensitivity of the rounded vector is bounded by
\begin{equation}
\label{eq:discrete_gaussian_sensitivity}
    \tilde{\Delta}_2(f)
    \le
    \min\left\{
        \Delta_2(f) + h\sqrt{d},\,
        \sqrt{\vphantom{\frac{\sqrt{d}h^2}{2}}
            \Delta_2^2(f)
            + \frac{dh^2}{4}
            + \sqrt{2\ln(1/\beta)}
            \left(
                \Delta_2(f)h + \frac{\sqrt{d}h^2}{2}
            \right)
        }
    \right\}.
\end{equation}
To guarantee privacy, the rounding must be repeated until the required condition on the $\ell_2$ norm is satisfied.
(This is the case with high probability, so the privacy cost of the rounding is low.)

{\bf Randomness complexity of the discrete Gaussian.}
The choice of discretization parameter $h$ creates a tension for the discrete Gaussian mechanism because the noise must be generated on the lattice $h\mathbb{Z}^d$.
Equivalently, one may first scale the function by $1/h$, round and apply the discrete Gaussian mechanism on $\mathbb{Z}^d$, and then rescale the output by multiplying with $h$.
If $\sigma > 0$ is the noise level required by the Gaussian mechanism for differential privacy, then the corresponding discrete mechanism uses scale $\tau = \sigma \Delta_2(\tilde f)/h$, before the resulting noise is ultimately rescaled by $h$.
The variance of the resulting mechanism is close to that of the continuous Gaussian mechanism provided that $h \ll \Delta_2(f)/\sqrt{d}$.
However, even in the regime where the variance increase is small, the entropy increases significantly as $h$ approaches $0$.
In particular, the binary entropy of $Z_j$ sampled from the one-dimensional discrete Gaussian distribution with scale $\tau$ is
\begin{equation}
H(Z_j)=\frac{1}{2}\log_2(2\pi e\,\tau^2)+o_\tau(1),
\end{equation}
see Lemma~6 of \citet{ling2014achieving}.
Since the coordinates of $Z=(Z_1,\dots,Z_d)$ are independent, it follows that the binary entropy of the $d$-dimensional discrete Gaussian is
\begin{equation}
H(Z)=\sum_{j=1}^d H(Z_j)=\frac{d}{2}\log_2(2\pi e\,\tau^2)+o_\tau(d).
\end{equation}
With $\tau = \sigma \Delta_2(\tilde f)/h$ this is approximately
\(
d\log_2\!\left(\frac{\sigma \Delta_2(\tilde f)}{h}\right)
+ 2d
\)
for large enough \(\tau\).

\section{Dithered Gaussian Mechanism}

We now present an alternative to the discrete Gaussian that avoids the need to discretize the input.
Unlike the methods described in Section~\ref{sec:background}, it does not incur any increase in sensitivity.
Moreover, our method uses provably fewer random bits, making it practical even when noise must be generated using slow cryptographically secure randomness.
In addition, the method directly inherits the privacy guarantees of the (continuous) Gaussian mechanism, and therefore does not require separate proofs for composition or amplification by subsampling.

{\bf Conceptual description.}
Our method, the \emph{dithered Gaussian mechanism}, is illustrated in Figure~\ref{fig:dithered-rounding-geometry}.
It can be thought of as a post-processing of the Gaussian mechanism with noise scale \(\sigma\), i.e., the mechanism that outputs \(f(X)+y\), where \(y\sim\mathcal{N}(0,\sigma^2 I_d)\).
For a parameter \(\xi>0\) the post-processing simply rounds each coordinate to an axis-aligned grid of points with distance \(\xi\) between consecutive points along all axes.
Performing this rounding deterministically would introduce bias, but we avoid this with a random shift of the grid: 
Sample public randomness \((a,b)\sim \mathrm{Uniform}([0,1)^2)\) and define the coordinate-dependent dither 
 by \(\gamma_i=(ai+b)\bmod 1\) for each \(i\in[d]\).
Given \(f(X)\in\mathbb{R}^d\) and \(y\sim\mathcal{N}(0,\sigma^2 I_d)\), the mechanism outputs \(\M(f(X))\) where, for \(i=1,\dots,d\),
\begin{equation}\label{eq:mechanism}
\M(f(X))_i
=
\xi\left(
\left\lfloor
\frac{f(X)_i+y_i}{\xi}-\gamma_i+\frac12
\right\rfloor
+\gamma_i
\right).
\end{equation}

\begin{figure}[t]
    \centering
    \begin{tikzpicture}[
        x=0.9cm,
        y=0.9cm,
        dot/.style={circle, inner sep=0pt, minimum size=3pt},
        every node/.style={font=\large}
    ]
        \definecolor{ditherolive}{RGB}{45,47,0}
        \definecolor{mechanismorange}{RGB}{210,82,24}
        \definecolor{noisygreen}{RGB}{0,140,35}

        \foreach \x in {1,2,3,4,5} {%
            \foreach \y in {1,2,3,5} {%
                \node[dot, fill=black] at (\x,\y) {};
            }
            \node[dot, fill=ditherolive] at (\x,4) {};
        }

        \draw[
            decorate,
            decoration={brace, amplitude=6pt},
            line width=0.8pt
        ] (0.8,4) -- (0.8,5)
        node[midway, left=8pt] {$\xi$};

        \node[dot, fill=noisygreen, minimum size=4pt] (noisy) at (2.75,2.45) {};
        \node[dot, fill=mechanismorange, minimum size=4pt] (rounded) at (3,2) {};
        \draw[->, line width=0.8pt] (noisy) -- (rounded);

        \node[
            noisygreen,
            font=\small,
            anchor=west
        ] at (1.0,2.5) {$f(X)+y$};
        \node[
            mechanismorange,
            font=\small,
            anchor=west
        ] at (2.2,1.65) {$\mathcal{M}(f(X))$};
        \node[anchor=west] at (3.0,4) {$\xi\gamma$};
    \end{tikzpicture}
    \caption{From a privacy perspective, the dithered Gaussian mechanism can be thought of as mapping the Gaussian mechanism output \(f(X)+y\) to the nearest point on a randomly shifted axis-aligned grid \(\{\xi (\gamma + z)\; | \; z\in\mathbb{Z}^d\}\). From an implementation perspective, we directly sample from a discrete distribution over the grid.}
    \label{fig:dithered-rounding-geometry}
\end{figure}
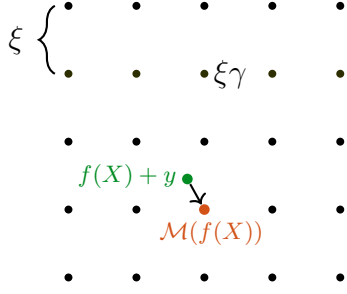

{\bf Direct sampling from a discrete distribution.}
To avoid generating the Gaussian vector \(y\) we observe that it is possible to sample the integer-valued random variable
\[
Z_i := \left\lfloor \frac{f(X)_i+y_i}{\xi}-\gamma_i+\frac12 \right\rfloor
\]
directly without first sampling \(y_i\).
Indeed, for each \(k\in\mathbb Z\), the event \(Z_i=k\) is equivalent to
\[
k \le \frac{f(X)_i+y_i}{\xi}-\gamma_i+\tfrac12 < k+1  \Leftrightarrow  \frac{\xi\left(k+\gamma_i - \frac12\right)-f(X)_i}{\sigma} \leq \frac{y_i}{\sigma} < \frac{\xi\left(k+\gamma_i + \frac12\right)-f(X)_i}{\sigma}.
\]
Denoting the cumulative distribution function of \(\frac{y_i}{\sigma}\sim\mathcal{N}(0,1)\) by \(\Phi\) we thus have:
\begin{equation}
\label{eq:Z_k}
\Pr[Z_i=k]
=
\Phi\!\left(
\frac{\xi\left(k+\gamma_i+\frac12\right)-f(X)_i}{\sigma}
\right)
-
\Phi\!\left(
\frac{\xi\left(k+\gamma_i-\frac12\right)-f(X)_i}{\sigma}
\right),
\qquad k\in\mathbb Z.
\end{equation}
Since \(\mathcal M(f(X))_i=\xi(Z_i+\gamma_i)\), this shows that to compute \(\mathcal M(f(X))_i\) it suffices to sample directly from a discrete distribution over \(\mathbb Z\). The crucial property of the proposed distribution is that its cumulative probabilities coincide with the Gaussian distribution function evaluated at the corresponding quantization boundaries. Indeed, the differences in \eqref{eq:Z_k} telescope, giving
\begin{equation}
\Pr[Z_i\le k]
=
\Phi\!\left(
\frac{\xi\left(k+\gamma_i+\frac12\right)-f(X)_i}{\sigma}
\right).
\end{equation}
Consequently, inverse-transform sampling reduces to evaluating the inverse Gaussian distribution function $\Phi^{-1}$, allowing us to sample $Z_i$ accurately and efficiently. Note that this convenient representation is not available for the discrete Gaussian distribution, whose cumulative probabilities involve a normalized sum of Gaussian weights over the integer lattice.

\begin{algorithm}[t]
\caption{Dithered Gaussian mechanism}
\label{alg:dithered-gaussian}
\begin{algorithmic}[1]
\Require Function value $f(X) \in \mathbb{R}^d$, parameters $\xi > 0$, $\sigma > 0$
\State Sample public randomness $(a,b) \sim \mathrm{Uniform}([0,1)^2)$
\For{$i = 1, \dots, d$}
    \State $\gamma_i \gets (a \cdot i + b) \bmod 1$
    \State Sample an integer $Z_i \in \mathbb{Z}$ such that, for every $k \in \mathbb{Z}$,
    \[
    \Pr[Z_i = k]
    =
    \Phi\!\left(
    \frac{\xi\left(k+\gamma_i+\frac12\right)-f(X)_i}{\sigma}
    \right)
    -
    \Phi\!\left(
    \frac{\xi\left(k+\gamma_i-\frac12\right)-f(X)_i}{\sigma}
    \right)
    \]
    \State $\mathcal{M}(f(X))_i \gets \xi (Z_i + \gamma_i)$
\EndFor
\State \Return $\mathcal{M}(f(X))$
\end{algorithmic}
\end{algorithm}

{\bf Algorithm details.}
We summarize the mechanism and its sampling distribution in Algorithm~\ref{alg:dithered-gaussian}.
Conceptually, $Z_i$ can be sampled by inverse transform sampling using ideal, real-valued arithmetic, we outline the procedure next.
Instead of drawing an infinite-precision uniform variable in $[0,1)$, we
progressively generate random bits.
After finitely many bits, the sampled prefix determines an interval
$[L_i,R_i)\subset[0,1)$ containing the underlying uniform variable.
We use the midpoint of this interval to obtain a candidate $k_i$ and compute
its two corresponding probability boundaries.
If the entire interval $[L_i,R_i)$ lies between these boundaries, then all
possible continuations of the random bit sequence produce the same outcome,
and we return $Z_i=k_i$.
Otherwise, additional random bits are generated and the interval is refined.
In practice, the procedure is implemented in a vectorizable form in PyTorch.
Most coordinates are resolved using float64 arithmetic.
Coordinates that are too close to a probability boundary, or for which the float64 computation is not sufficiently reliable, are instead evaluated using arbitrary-precision arithmetic provided by the Python library \texttt{mpmath}.
Additional random bits are generated only for the remaining active coordinates.
This makes the procedure efficient on GPUs while avoiding an explicit truncation of the support.
Details can be found in Appendix~\ref{app:sampling}. We provide an implementation of Algorithm~\ref{alg:dithered-gaussian} in our GitHub repository.\footnote{\url{https://github.com/npkalinin/DP-Dithered-Gaussian-Mechanism}}

{\bf Noise distribution.}
We characterize the distribution of \(\M(X)\) over the random choice of $\gamma$ and $Z$:

\begin{restatable}{lemma}{ditheredGaussianRoundingLemma}
\label{lem:full_noise_dithered_gaussian}
    Over the randomness of the shift $\gamma_i$ and index $Z_i$, $\mathcal{M}(f(X))_i$ is identically distributed to
\begin{equation}
     f(X)_i + y_i + u_i,
\end{equation}
where $u_i \sim \mathrm{Uniform}(-\xi/2,\xi/2)$ and \(y_i\sim\mathcal{N}(0,\sigma^2)\) are independent.
\end{restatable}
The lemma is a consequence of the fact that the grid shift $\gamma_i$ is uniformly distributed and independent of the Gaussian noise $y_i$, which makes the perturbation caused by rounding to the nearest grid point follow a uniform distribution on $(-\xi/2,\xi/2)$. The full proof is provided in Appendix~\ref{app:proofs}.

From a utility perspective, the dithered Gaussian mechanism is therefore equivalent to adding Gaussian noise $\mathcal{N}(0,\sigma^2)$ followed by uniform noise $\mathrm{Uniform}(-\xi/2,\xi/2)$.
The variance of the resulting noise is
\begin{equation}
    \mathrm{Var}(y_i + u_i) = \sigma^2 + \xi^2/12.
\end{equation}

We can also quantify directly how close this noise distribution is to the Gaussian distribution used by the standard Gaussian mechanism in terms of the total variation distance.

\begin{restatable}{proposition}{ditheredGaussianTVBound}
\label{prop:dithered_gaussian_tv_bound}
Let \(G_\sigma\sim\mathcal N(0,\sigma^2)\), and let \(U_\xi\sim\mathrm{Uniform}(-\xi/2,\xi/2)\) be independent of \(G_\sigma\).
Then
\begin{equation}
\label{eq:dithered_gaussian_tv_bound}
    \dTV(G_\sigma+U_\xi,G_\sigma)
    \le
    0.0202
    \left(\frac{\xi}{\sigma}\right)^2.
\end{equation}
\end{restatable}

Thus the distributional error introduced by dithering is quadratic in the relative grid width \(\xi/\sigma\).

Densities for dithered Gaussians for various choices of $\xi$ can be seen in Figure~\ref{fig:dithered-gaussian-distribution}.
For $\xi = \sigma$ we see that the density is nearly indistinguishable from that of the standard Gaussian, while larger values of~$\xi$ somewhat flatten the distribution.

\begin{figure}
    \centering
    \includegraphics[width=0.9\linewidth]{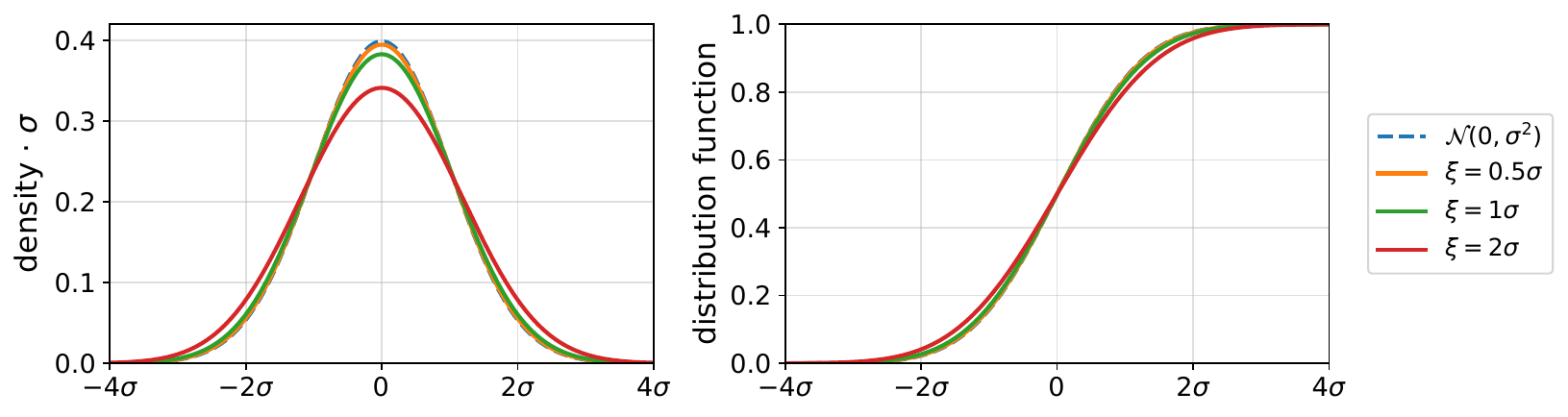}
    \caption{Density and cumulative distribution function of the noise distribution of the dithered Gaussian mechanism for various choices of $\xi$, compared to the Gaussian mechanism with the same privacy guarantee. Note that $\xi=\sigma$ provides a good approximation of the Gaussian distribution.}
    \label{fig:dithered-gaussian-distribution}
\end{figure}

{\bf Randomness complexity.}
Next, we analyze the randomness complexity of the dithered Gaussian mechanism.
We distinguish between the private (high-quality) randomness used for discrete sampling and the public randomness used to generate the pair $(a, b)$ of uniform random variables defining the grid shift $\gamma$.
The value $\gamma$ may be shared with the adversary without compromising privacy, its purpose is only to ensure unbiasedness, and to improve utility and expected randomness complexity.

We bound the entropy of the dithered Gaussian mechanism in the following lemma:

\begin{restatable}{lemma}{ditheredGaussianEntropyBound}
\label{lem:dithered_gaussian_entropy_bound}
The worst-case private binary entropy of the dithered Gaussian mechanism is bounded by
\begin{equation}
    H(Z\mid \gamma)
    \le
    \frac{d}{2}\log_2\left(
        2\pi e\left[
            \left(\frac{\sigma}{\xi}+\frac12\right)^2+\frac1{12}
        \right]
    \right).
\end{equation}
Moreover, as $\sigma/\xi\to 0$, the following tighter conditional entropy bound holds:
\begin{equation}
    \mathbb{E}_{\gamma}\left[H(Z\mid\gamma)\right]
    = O\left(\frac{d\sigma}{\xi}\log\frac{\xi}{\sigma}\right).
\end{equation}
\end{restatable}
The full proof is given in  Appendix~\ref{app:proofs}.

As noted above one can choose $\xi \approx \sigma$ and get an error distribution that is close to Gaussian.
In this case, the lemma says that the entropy per coordinate is a small constant.
Intuitively, the mechanism has low entropy for \(\xi = \Omega(\sigma)\) since the random noise \(y_i\) is unlikely to make \(\M(f(X))_i\) differ significantly from the output with zero noise, \(y_i = 0\).
Thus most probability mass for each coordinate is on a few output values, resulting in low entropy.
For \(\xi \gg \sigma\) the output distribution is almost deterministic for most values of \(\gamma_i\), resulting in even smaller expected entropy.
This argument extends and refines similar observations by \citet{ghentiyala2026}.

In theory, using techniques for sampling from low-entropy distributions based on explicit, ordered sampling probabilities~\citep{knuth1976nonuniform,devroye1986nonuniform}, multiple runs of the mechanism can be implemented in time near-linear in \(d\) using an expected number of random bits that matches the sum of all private entropy bounds plus \(O(\log(1/\beta))\) bits with probability \(1-\beta\).
As mentioned above, details of the implementation of our actual sampling mechanism can be found in Appendix~\ref{app:sampling}.

\paragraph{Lower Bound.}
In Appendix~\ref{app:lower_bound}, we prove the following lower bound on the private entropy of any differentially private mechanism for releasing a sensitivity-$C$ statistic.
To make the binary entropy well-defined we consider the special case where the statistic is integer-valued.
Inputs \(z\) and \(z'\) are considered neighboring if \(|z-z'|\le C\) where \(C\in\mathbb{N}\).
Let \(U\) denote the public randomness, and write \(\M_U\) for the mechanism obtained after conditioning on \(U\).
\begin{restatable}{theorem}{mainlower}\label{thm:main-lower}
There are absolute constants \(c,c_0>0\) such that for \(\varepsilon\in(0,1)\) and \(\delta\le c_0\varepsilon\) the following holds.
Suppose $\M_U: \ZZ \to \ZZ$ is an $(\varepsilon,\delta)$-differentially private mechanism under sensitivity \(C\), for every fixed value of public randomness \(U\).
If, for some \(\alpha\in\mathbb{N}\) and \(\beta\in(0,\tfrac{1}{10})\), for every \(z\in\ZZ\) it holds that \(\Pr_{U,R}\!\left[|\M_U(z)-z|>\alpha\right]\le \beta\) (where \(R\) is the private randomness) then there exists \(z\in \ZZ\) such that
\[
 \E_U\!\big[\,H(\M_U(z))\,\big] \ge c\,\frac{C}{\varepsilon\,\alpha}.
\]
\end{restatable}
The proof proceeds by fixing a suitable value of the public randomness and using a graph-theoretic argument based on component stability and barrier structure to derive a lower bound on the achievable accuracy.
Since the Gaussian noise scale needed for privacy is proportional to the clipping radius \(C\) (up to the usual dependence on \(\varepsilon,\delta\)), this lower bound shows that the relevant scale for private entropy is the ratio between the privacy noise scale and the accuracy scale.
This is consistent with Lemma~\ref{lem:dithered_gaussian_entropy_bound}: the private entropy of the dithered Gaussian mechanism depends on the dimensionless ratio \(\sigma/\xi\), where \(\sigma\) is the Gaussian noise scale and \(\xi\) is the discretization width.
In particular, choosing \(\xi=\Theta(\sigma)\) gives constant private entropy per coordinate, while taking a finer grid increases the entropy only logarithmically in \(\sigma/\xi\).

\begin{figure}[t]
    \centering
    \begin{subfigure}[t]{0.45\linewidth}
        \centering
        \includegraphics[width=\linewidth]{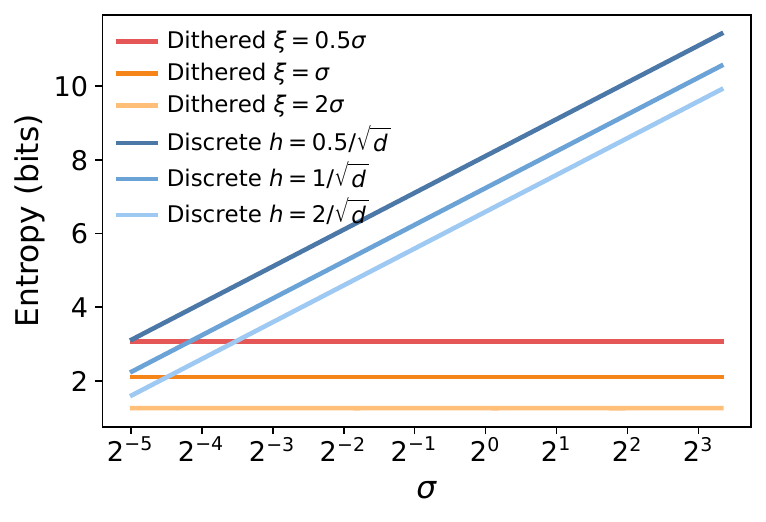}
        \caption{Entropy in bits.}
        \label{fig:entropy_discrete}
    \end{subfigure}
    \hfill
    \begin{subfigure}[t]{0.45\linewidth}
        \centering
        \includegraphics[width=\linewidth]{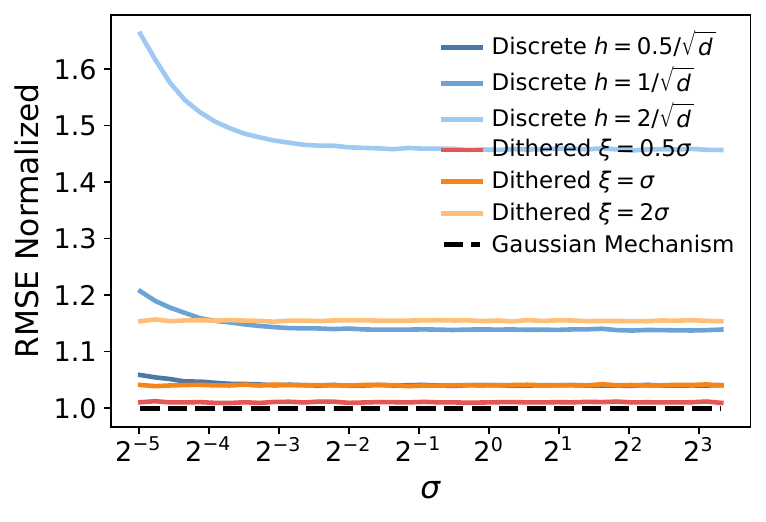}
        \caption{Normalized RMSE.}
        \label{fig:rmse_discrete}
    \end{subfigure}
    \caption{Comparison of the dithered and discrete Gaussian mechanisms as a function of the noise scale $\sigma$.
The left panel reports the entropy per coordinate.
The right panel reports the RMSE on vectors in dimension $d=1000$ sampled uniformly from the unit sphere.
For each value of $\sigma$, we run each mechanism $1000$ times and compute the average reconstruction error.
The RMSE is normalized by the Gaussian mechanism baseline, which is equivalent to dividing by $\sigma$.
The proposed dithered Gaussian mechanism requires a much smaller number of bits, which is independent of $\sigma$.}
    \label{fig:discrete_comparison}
\end{figure}

\paragraph{Dithered Laplace mechanism.}

The proposed dithering scheme also extends to the Laplace mechanism. Appendix~\ref{sec:dithered_laplace} gives the corresponding analysis, including an improved randomness complexity bound over~\citet{canonne2025randomness}; see Theorem~\ref{thm:main_dithered_laplace}.

\section{Experiments}


{\bf Entropy and RMSE.}
Assume that the sensitive vectors are drawn from the unit sphere in dimension $d = 1000$. In the case of the discrete Gaussian mechanism, these vectors are further discretized with discretization levels $h \in \{\frac{0.5}{\sqrt{d}}, \frac{1}{\sqrt{d}}, \frac{2}{\sqrt{d}}\}$, roughly preserving vector norms.
We set the resampling parameter to $\beta = e^{-1/2}$, as suggested by~\citet{kairouz2021distributed}.
For the dithered Gaussian mechanism, we vary $\xi \in \{0.5\sigma, \sigma, 2\sigma\}$.
In Figure~\ref{fig:discrete_comparison}, we show that, for a sufficiently large noise level $\sigma$, our method has much lower entropy and therefore requires fewer random bits to sample. In terms of utility, measured by the root mean squared error normalized by $\sigma$, the dithered Gaussian mechanism is much closer to the Gaussian mechanism, partly because it does not require discretizing the vectors.

{\bf Time Comparison.}
We compare the runtime of the dithered and discrete Gaussian mechanisms for dimensions $d\in\{10^3,10^4,10^5,10^6,10^7\}$, using the same parameter choices as above. Our main baseline is the discrete Gaussian implementation in \texttt{jax\_privacy}~\citep{jax_privacy2022github}, the fastest implementation we found. Across the considered hardware and randomness settings, the dithered Gaussian mechanism is between $2\times$ and $500\times$ faster, with the largest gains in high dimensions; see Tables~\ref{tab:sampling_runtime_cpu_pseudo}--\ref{tab:sampling_runtime_cuda_secure}. We additionally consider discrete inputs in $[-100,100]$, a setting favorable to the discrete Gaussian mechanism, and still observe speedups of up to $10\times$ for $d\geq10^4$; see Table~\ref{tab:sampling_runtime_discrete}. Full experimental and implementation details are provided in Appendix~\ref{sec:experimental_details}.

\begin{figure}[t]
    \centering
    \begin{subfigure}{0.38\linewidth}
        \centering
        \includegraphics[width=\linewidth]{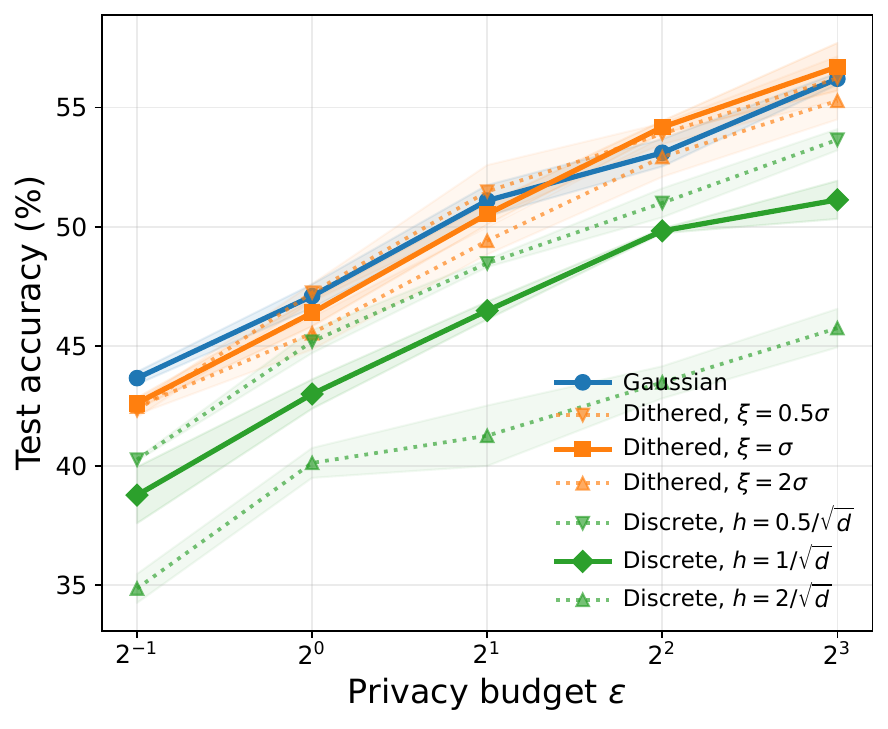}
        \caption{Test accuracy on CIFAR-10.}
        \label{fig:cifar10_accuracy}
    \end{subfigure}
    \hspace{2cm}
    \begin{subfigure}{0.38\linewidth}
        \centering
        \includegraphics[width=\linewidth]{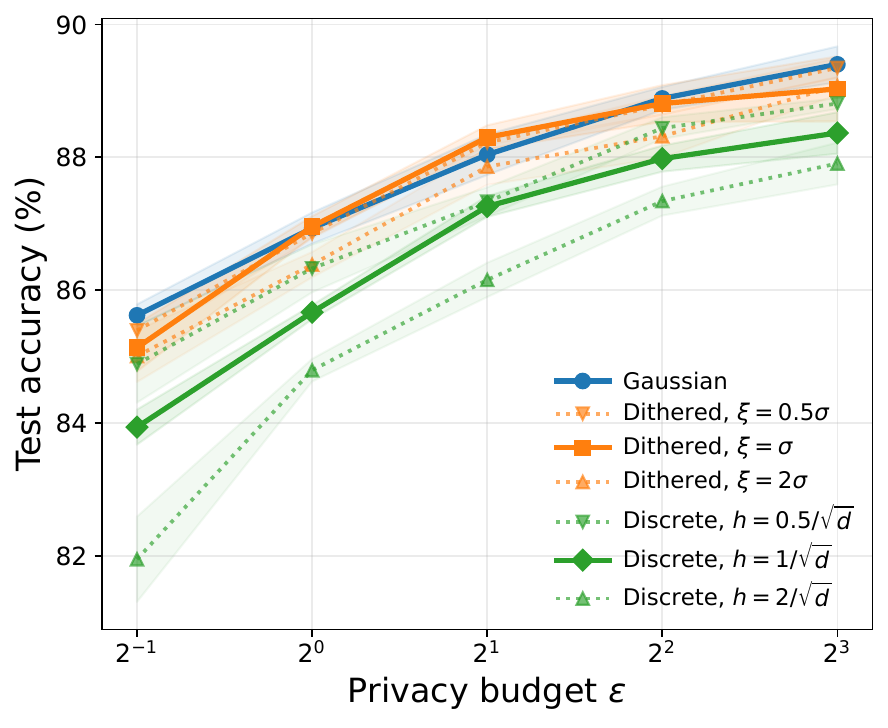}
        \caption{Test accuracy on IMDb.}
        \label{fig:imdb_accuracy}
    \end{subfigure}

    \caption{%
       Comparison of Gaussian, discrete and dithered Gaussian mechanisms across privacy budgets $\varepsilon$ with $\delta = 10^{-5}$ for ConvNet training on CIFAR-10 and BERT-Base fine-tuning on IMDb.
    }
    \label{fig:accuracy_vs_epsilon}
\end{figure}

{\bf Model Training\@.}
We next evaluate the effect of the dithered Gaussian mechanism in an end-to-end DP-SGD training. We train a ConvNet on CIFAR-10 and fine-tune BERT-Base on the IMDb dataset using Opacus~\citep{yousefpour2021opacus} for privacy budgets $\varepsilon \in \{0.5,1,2,4,8\}$ with $\delta=10^{-5}$, and compare the Gaussian, discrete Gaussian, and our dithered Gaussian mechanism. As shown in Figure~\ref{fig:accuracy_vs_epsilon} and Tables~\ref{tab:cifar10_accuracy_privacy} and ~\ref{tab:imdb_bert_accuracy_privacy}, the dithered Gaussian mechanism closely matches the accuracy of the Gaussian baseline. The accuracy degradation of the discrete Gaussian mechanism is more pronounced, although it is reduced by using a finer discretization grid.

The runtime comparisons in Tables~\ref{tab:cifar10_runtime_privacy} and ~\ref{tab:imdb_bert_runtime_privacy} shows that the dithered Gaussian mechanism with cryptographically secure noise  incurs only a $3\%$--$7\%$ runtime overhead relative to standard Gaussian DP-SGD for batch size $512$. For the discrete Gaussian mechanism, we follow the implementation in \texttt{jax\_privacy} and use rounding to neighboring integers, without random rotations, for computational efficiency.
Nevertheless, this implementation still incurs approximately a $45\%$ runtime overhead for ConvNet and $320\%$ for BERT-Base.
We also compare with Opacus' ``secure mode'', which is slightly slower than our proposed mechanism, mainly because it requires sampling a larger number of random bytes. However, in terms of protection against floating-point vulnerabilities and privacy accounting, secure mode should be viewed as a heuristic rather than as providing a formal guarantee. Further details on the experimental setup, privacy accounting, and implementations are provided in Appendix~\ref{sec:experimental_details}.

The overhead is indicative of what one can expect in general, though the figure will depend on how much time is spent on randomness generation, which in turn depends on factors such as batch size. For instance, even when using only cryptographically secure noise, much larger overheads were observed in the work of~\citet{egan2024highspeed}, ranging from $74\%$ to over $400\%$. In that work, a rather small batch size of $32$ is used, making the time spent on noise generation relatively larger. It should therefore be expected that, with a much larger batch size, the overhead would drop.

\section{Conclusion and Future Directions}

We have introduced the dithered Gaussian mechanism, an alternative to previous discrete noise mechanisms, which protects against floating-point vulnerabilities without requiring separate privacy accounting.
Instead, it directly inherits the privacy guarantees of the Gaussian mechanism via post-processing. 
We prove that it requires fewer random bits than the discrete Gaussian mechanism and only moderately increases the time needed to sample the noise compared with a naive, non-secure implementation.

A promising direction for future work is to generalize the proposed dithered Gaussian mechanism to the distributed differential privacy setting by integrating it with secure aggregation in federated learning.
The main obstacle is that the dithered Gaussian mechanism is not additive, which prevents its noise from being straightforwardly decomposed across clients.
A separate vulnerability, shared by all of the methods considered here, is exposure to timing attacks.
We also note that the mechanism would likely benefit from dedicated hardware support to directly represent vectors in discretized form, rather than converting to a floating point representation, and to speed up the sampling step.
Conceivably this could make cryptographically secure randomness available in DP-SGD with a truly negligible performance overhead.

{\bf Acknowledgements.}

We thank Thomas Steinke, Salil Vadhan, and Hossein Zakerinia for their helpful comments on an earlier version of the paper. We also thank Ryan McKenna for a fruitful discussion on the discrete Gaussian mechanism.

We thank Andreas V.~Welsch Zacchi and
Christoffer H.~Andersen who independently obtained empirical results on DP-SGD supporting our conclusions as part of their BSc thesis at the University of Copenhagen.

Rasmus Pagh is supported by a Data Science Distinguished Investigator grant from Novo Nordisk Fonden, and is part of BARC, supported by the VILLUM Foundation grant~54451.
Nikita Kalinin is supported in part by the Austrian Science Fund (FWF) [10.55776/COE12].

\bibliography{iclr2027_conference}
\bibliographystyle{iclr2027_conference}

\appendix

\clearpage
\section{Differential Privacy Basics}\label{app:dp-basics}

In this appendix we recall the standard definitions of differential privacy, see~\citet{dwork2014algorithmic} for details and a historical overview.
We adopt the standard \emph{unbounded} (add/remove-one) adjacency model: two datasets $X,X'\subseteq\mathbb{X}$ are said to be \emph{neighbors}, written $X\sim X'$, if one can be obtained from the other by either adding or removing a single element.
Formally,
\[
X\sim X' \quad\Longleftrightarrow\quad |\,X\triangle X'\,| = 1,
\]
where $\triangle$ denotes the symmetric difference.
This model corresponds to protecting the participation of a single individual, and is the default unless otherwise stated.
\begin{definition}[Pure and approximate differential privacy]
Let $\varepsilon\ge 0$ and $\delta\in[0,1)$.
A randomized mechanism $\M$ with range $\mathcal{Y}$ is \emph{$(\varepsilon,\delta)$-differentially private} if for all neighboring datasets $X\sim X'$ and all measurable $S\subseteq\mathcal{Y}$,
\[
\Pr[\M(X)\in S]\ \le e^{\varepsilon}\,\Pr[\M(X')\in S]\ +\ \delta.
\]
When $\delta=0$ we say that $\M$ is \emph{$\varepsilon$-differentially private} (\emph{pure} differential privacy).
\end{definition}

\begin{lemma}[Post-processing]
If a mechanism $\M: \mathbb{X}^*\to\mathcal{Y}$ is $(\varepsilon,\delta)$-differentially private and $K$ is any (possibly randomized) mapping $K:\mathcal{Y}\to\mathcal{Z}$ whose internal randomness is independent of the input dataset $X$, then the composed mechanism $K\circ\M$ is also $(\varepsilon,\delta)$-differentially private.
\end{lemma}

\begin{definition}[$\ell_p$-sensitivity]
For a function $f:\mathbb{X}^*\to\mathbb{R}^d$ and $p\in[1,\infty]$, the $\ell_p$-sensitivity is
\[
\Delta_p(f)\ :=\ \sup_{X\sim X'}\ \|f(X)-f(X')\|_p.
\]
\end{definition}

\begin{definition}[Laplace mechanism]\label{def:laplace-mechanism}
  For $\lambda >0$ and a function $f: \mathbb{X}^* \rightarrow \mathbb{R}^d$, the \emph{Laplace mechanism} is defined as
\(
\M(X)\ =\ f(X)\ +\ y\) where
\(y_i\sim \mathrm{Laplace}(\lambda)\) independently for $i=1,\dots,d$.
\end{definition}

The Gaussian mechanism adds isotropic Gaussian noise to a vector-valued function.
Namely, for $\sigma>0$ and $f:\mathbb{X}^*\to\mathbb{R}^d$, it is given by
$\M(X)=f(X)+z$, where $z\sim\mathcal{N}(0,\sigma^2 I_d)$.
The following result of~\citet{balle2018improving} gives an exact
characterization of its privacy guarantee in terms of the $\ell_2$-sensitivity
of $f$.

\begin{theorem}[Analytic Gaussian mechanism]\label{thm:gaussian-mechanism}
Let $f:\mathbb{X}^*\to\mathbb{R}^d$ have $\ell_2$-sensitivity
$\Delta_2(f)>0$, and let $\M(X)=f(X)+z$ with
$z\sim\mathcal{N}(0,\sigma^2 I_d)$. Then, for $\varepsilon\ge 0$ and
$\delta\in[0,1)$, the mechanism $\M$ is $(\varepsilon,\delta)$-differentially
private if and only if
\[
\delta\ \ge\
\Phi\!\left(
    \frac{\Delta_2(f)}{2\sigma}
    - \frac{\varepsilon\sigma}{\Delta_2(f)}
\right)
-
e^\varepsilon
\Phi\!\left(
    -\frac{\Delta_2(f)}{2\sigma}
    - \frac{\varepsilon\sigma}{\Delta_2(f)}
\right),
\]
where $\Phi$ denotes the cumulative distribution function of the standard
Gaussian distribution.
\end{theorem}

\section{Sampling Algorithm}\label{app:sampling}

\begin{algorithm}[t!]
\caption{Dithered Gaussian Sampling}
\label{alg:adaptive-dithered-gaussian}
\begin{algorithmic}[1]
\Require $f(X)_i\in\mathbb{R}$ and $\gamma_i\in[0,1)$ for $i\in[d]$, parameters $\xi>0$, $\sigma>0$, block size $t\ge 1$
\Ensure Samples $Z_1,\dots,Z_d$

\State Initialize active set $A\gets[d]$; for each $i$, no outcome has yet been assigned

\While{$A\neq\emptyset$}
    \State Generate a block $B_i\in\{0,\dots,2^t-1\}$ of $t$ random bits for each $i\in A$
    \For{each $i\in A$}
        \State Let $B_{i,1},\dots,B_{i,r_i}$ be all blocks generated so far for row $i$
        \State $L_i\gets\sum_{s=1}^{r_i}B_{i,s}2^{-ts}$, \quad
               $R_i\gets L_i+2^{-tr_i}$
        \State $U_i\gets(L_i+R_i)/2$
        \State $k_i\gets
        \left\lfloor
        \dfrac{f(X)_i+\sigma\Phi^{-1}(U_i)}{\xi}
        -\gamma_i+\tfrac12
        \right\rfloor$
        \State $F_i(k_i-1)\gets
        \Phi\!\left(
        \dfrac{\xi(k_i+\gamma_i-\tfrac12)-f(X)_i}{\sigma}
        \right)$
        \State $F_i(k_i)\gets
        \Phi\!\left(
        \dfrac{\xi(k_i+\gamma_i+\tfrac12)-f(X)_i}{\sigma}
        \right)$
        \If{$[L_i,R_i)\subseteq[F_i(k_i-1),F_i(k_i))$}
            \State $Z_i\gets k_i$
            \State Remove $i$ from $A$
        \EndIf
    \EndFor
\EndWhile

\State \Return $Z_1,\dots,Z_d$
\end{algorithmic}
\end{algorithm}

In this section, we present Algorithm~\ref{alg:adaptive-dithered-gaussian} for sampling from the dithered Gaussian distribution (\eqref{eq:Z_k}). The critical property of the dithered Gaussian distribution is that, although it is discrete, its cumulative distribution function can be expressed directly in terms of the Gaussian CDF as $F_i(k)=\Phi\!\left((\xi(k+\gamma_i+\tfrac12)-f(X)_i)/\sigma\right)$. Consequently, for any $U_i\in(0,1)$, applying the Gaussian inverse CDF yields a value whose quantization matches the corresponding interval of the dithered distribution. In particular, given $U_i\sim\mathrm{Unif}(0,1)$, the candidate $k_i=\left\lceil (f(X)_i+\sigma\Phi^{-1}(U_i))/\xi-\gamma_i-\tfrac12\right\rceil$ is precisely the integer whose CDF interval contains $U_i$, that is, $U_i\in[F_i(k_i-1),F_i(k_i))$. Hence, $k_i$ is a valid sample from the desired dithered Gaussian distribution.

Sampling a true uniform random variable requires an infinite number of bits. Therefore, instead, we sample a finite number $r_i\times t$ of bits, structured into $r_i$ blocks of $t$ bits for efficiency, since sampling one extra bit at a time is impractical. The sampled bits determine an interval $[L_i,H_i)$ containing the uniform random variable that would be obtained if sampling were continued indefinitely. We use the midpoint $U_i=(L_i+H_i)/2$ to obtain a candidate $k_i$ through the inverse CDF. If the entire interval $[L_i,H_i)$ is contained in the corresponding CDF interval $[F_i(k_i-1),F_i(k_i))$, then all possible continuations of the sampled bits lead to the same value $k_i$, and we output $Z_i=k_i$. Otherwise, we sample another block of $t$ bits and refine the interval. We maintain a set of indices $A$ for which the sampled value has not yet been determined; we call this the active set. In practice, we choose the block size so that most values are determined on the first attempt, and additional bits are sampled only for the remaining, smaller active set.

\paragraph{Finite precision in practical implementation.}
Being implemented in PyTorch, the efficiency of the sampler depends on resolving most coordinates using float64 arithmetic. For each coordinate, we first use \texttt{torch.special.ndtri} to obtain a candidate value $k_i$ from the midpoint of the current random interval $[L_i,R_i)$, and then verify the candidate by evaluating the two corresponding CDF boundaries $F_i(k_i-1)$ and $F_i(k_i)$ using \texttt{torch.special.ndtr}. To account for finite-precision error, we only accept a value when the interval lies inside the corresponding CDF cell with a conservative margin $\epsilon=10^{-12}$, i.e., when $L_i\ge F_i(k_i-1)+\epsilon$ and $R_i\le F_i(k_i)-\epsilon$.  Coordinates that cannot be certified are kept in the active set and sampled with additional random bits, which shrinks $[L_i,R_i)$ while leaving the target distribution unchanged.

After a fixed number of such fast refinement steps, the small number of unresolved coordinates are handled using arbitrary-precision arithmetic with the Python library \texttt{mpmath}. The original float64 parameters are converted to their exact rational representations, and the CDF boundaries are recomputed at a precision that increases with the number of sampled random bits. If the current random interval is still too close to a CDF boundary to be resolved at that precision, another block of random bits is sampled and the procedure is repeated. Thus, unlike the truncated implementation, the practical sampler introduces no truncation of the support: finite precision only determines whether a coordinate can be safely resolved in the fast float64 path or must be deferred to the high-precision fallback.

To make the procedure efficient and prevent float64 overflow or loss of precision, we also normalize the vector $f(X)$ by prescaling it by the clipping norm before clipping, thereby guaranteeing that its coordinates are bounded by $1$ in absolute value. After secure discrete sampling, we rescale the vector by multiplying it by the clipping norm back. We may further require that the noise multiplier itself be represented by a reasonable float64 value, say $\sigma \in [10^{-5}, 10^2]$, to avoid loss of precision during the computation. These are, however, not binding assumptions, since we can always rely on arbitrary precision computation, albeit at the cost of reduced efficiency.

We provide an implementation of the proposed method in our GitHub repository

\url{https://github.com/npkalinin/DP-Dithered-Gaussian-Mechanism}

\section{Experimental Details}
\label{sec:experimental_details}

\paragraph{Randomness Efficiency.}
Our implementation of the dithered Gaussian mechanism takes $2$ bytes of randomness per coordinate on average. Most of the coordinates need even less than this, but we take this overconservative estimate to gain efficiency. An efficient implementation of the discrete Gaussian mechanism in \texttt{jax\_privacy}~\citep{jax_privacy2022github} takes on average $40$ bytes per coordinate, so $20$ times more randomness than we do. In particular, they sample 3 random variables of size $8$ bytes and use oversampling of size
\begin{equation}
    m = \bigg\lfloor\frac{5d}{3}\bigg\rfloor + 2 \lfloor \sqrt{d}\rfloor + 4,
\end{equation}
resulting in roughly $5/3 \times 24 = 40$ bytes per coordinate for sufficiently large $d$. Numerical experiments under the same settings as those presented in the tables confirm this finding.

\paragraph{Time Measurements.}
We measure the runtime required to privatize a $d$-dimensional vector with unit $\ell_2$-sensitivity, using a fixed Gaussian noise multiplier of $\sigma=0.01$. Our primary baseline is the discrete Gaussian mechanism implemented in \texttt{jax\_privacy}~\citep{jax_privacy2022github}, which is the most efficient implementation we found. A comparison with other, slower implementations, including OpenDP~\citep{opendp} and diffprivlib~\citep{diffprivlib}, is reported in Table~\ref{tab:sampling_runtime_comparison}. We do not include TensorFlow Privacy~\citep{tensorflow_privacy}, since, to the best of our knowledge, it provides only an approximate implementation of the discrete Gaussian distribution and does not support cryptographically secure noise generation.

For the discrete Gaussian mechanism, we use the random-rotation procedure of~\citet{kairouz2021distributed}, with sensitivity given by~\eqref{eq:discrete_gaussian_sensitivity} and resampling parameter $\beta = e^{-1/2}$. This procedure additionally requires padding the input dimension to the next power of two and applying a random rotation before discretization, but in return gives lower error in the privatized vector.

We consider dimensions
$d \in \{10^3,10^4,10^5,10^6,10^7\}$,
dithering parameters
$\xi \in \{0.5\sigma,\sigma,2\sigma\}$,
and discrete-Gaussian discretization levels
$h \in \{0.5/\sqrt{d},1/\sqrt{d},2/\sqrt{d}\}$.
For both mechanisms, we consider cryptographically secure randomness generated using the Python standard-library \texttt{secrets} module, as well as faster PyTorch non-cryptographic pseudorandom generation. Our implementation of the dithered Gaussian mechanism is designed to run efficiently on GPUs. Implementations of the discrete Gaussian distribution are either designed primarily for CPUs or, as in the case of \texttt{jax\_privacy}, benefit less from GPU acceleration because the main bottleneck is random-number generation, which is performed on the CPU.  Overall, the dithered Gaussian mechanism is approximately $2$--$500\times$ faster than the \texttt{jax\_privacy} discrete Gaussian implementation across the tested settings, with the largest speedups occurring in high dimensions.

Tables~\ref{tab:sampling_runtime_cpu_pseudo}, \ref{tab:sampling_runtime_cuda}, \ref{tab:sampling_runtime_cpu_secure}, and~\ref{tab:sampling_runtime_cuda_secure} report the resulting noise-generation runtimes under the different hardware and randomness settings. CPU experiments use 16 cores of an AMD EPYC 9654 processor, while GPU experiments use a single NVIDIA A100 GPU together with 16 CPU cores. All reported runtimes are averaged over 100 runs.

We also report a runtime comparison for discrete inputs in the range $[-100,100]$ with noise multiplier $\sigma = 10$. This setting favors the discrete Gaussian mechanism, as it no longer requires an additional discretization step. The results are reported in Table~\ref{tab:sampling_runtime_discrete}. In particular, for sufficiently large dimensions, $d \geq 10^4$, our dithered Gaussian mechanism is still significantly faster than the considered implementations of the discrete Gaussian mechanism.

\paragraph{Model Training Experiments.}
We compare the Gaussian, discrete Gaussian, and proposed dithered Gaussian mechanisms in the setting of DP-SGD training. We use the Opacus-based~\citep{yousefpour2021opacus} implementation of DP-SGD to train a ConvNet ($d \approx 3\times10^5$) on CIFAR-10 and fine tune BERT-Base ($d \approx10^8$) on IMDb. In particular, for CIFAR-10 we use a batch size of $512$, a clipping norm of $1$, and train for $10$ epochs under privacy budgets $\varepsilon \in \{0.5,1,2,4,8\}$ with $\delta = 10^{-5}$. The optimal learning rate was the same across the $2\times$ search grid and equal to $0.002$ for all runs.
For IMDb experiments we use a batch size $1024$, a clipping norm of $1$, fine tune for $10$ epochs with the same privacy budgets. The optimal learning rate was $0.0005$. 
For the Gaussian and dithered Gaussian mechanisms, we use the PRV accountant provided by Opacus to account for subsampling. For the discrete Gaussian mechanism, following the implementation in \texttt{jax\_privacy}, we apply deterministic nearest-integer rounding to each gradient and use the zCDP analysis of the discrete Gaussian privacy profile together with the RDP accountant from the \texttt{dp\_accounting} library. This results in a slightly larger noise level than that required by the Gaussian mechanism.
We note that using the algorithm of~\citet{kairouz2021distributed} for the discrete Gaussian mechanism requires applying a random rotation and randomized rounding to each per-example gradient. In our experiments, this incurs approximately a $20\times$ runtime overhead, making this approach impractical for centralized learning.

We observe no significant accuracy degradation when using the dithered Gaussian mechanism; see Figure~\ref{fig:accuracy_vs_epsilon} and Tables~\ref{tab:cifar10_accuracy_privacy} and ~\ref{tab:imdb_bert_accuracy_privacy}. For the discrete Gaussian mechanism, the accuracy degradation is more pronounced, although it is reduced when using a finer discretization grid.
 
We report the runtime comparison~\ref{fig:imdb_accuracy} in Tables~\ref{tab:cifar10_runtime_privacy} and ~\ref{tab:imdb_bert_runtime_privacy}, using a single NVIDIA A100 GPU and 16 CPU cores. For noise generation, we use cryptographically secure randomness from the Python standard-library \texttt{secrets} module~\citep{python_secrets_2026}.
Our proposed dithered Gaussian mechanism incurs only a $3\%$--$7\%$ runtime overhead, whereas the discrete Gaussian mechanism is approximately $45\%$ slower for CIFAR-10 and $3.2$ times slower for IMDb experiments. We note that the current training implementation in \texttt{jax\_privacy} does not include this randomized rotation, resulting in a larger sensitivity but a faster training. We also compare against the ``secure mode'' of Opacus, which uses cryptographically secure noise generation and heuristics to reduce floating-point vulnerabilities. This mode is slower than our proposed mechanism, mainly because of the larger amount of randomness required.

\begin{table}[h!]
\centering
\caption{Sampling runtime on CPU in milliseconds with non-secure noise (mean $\pm$ standard deviation) over 100 runs.}
\label{tab:sampling_runtime_cpu_pseudo}
\begin{tabular}{lccccc}
\toprule
& $d=10^3$ & $d=10^4$ & $d=10^5$ & $d=10^6$ & $d=10^7$ \\
\midrule
Dithered ($\xi=0.5\sigma$)
& $0.35 \pm 0.19$
& $0.49 \pm 0.08$
& $1.51 \pm 0.06$
& $13.35 \pm 0.29$
& $\mathbf{105 \pm 4}$ \\

Dithered ($\xi=\sigma$)
& $0.17 \pm 0.07$
& $0.42 \pm 0.08$
& $\mathbf{1.49 \pm 0.04}$
& $\mathbf{13.25 \pm 0.24}$
& $127 \pm 19$ \\

Dithered ($\xi=2\sigma$)
& $\mathbf{0.14 \pm 0.04}$
& $\mathbf{0.37 \pm 0.11}$
& $1.50 \pm 0.04$
& $13.35 \pm 0.22$
& $131 \pm 19$ \\
\midrule
Discrete ($h=0.5/\sqrt{d}$)
& $0.39 \pm 0.01$
& $2.55 \pm 0.16$
& $12.94 \pm 0.91$
& $79.2 \pm 2.2$
& $1868 \pm 22$ \\

Discrete ($h=1/\sqrt{d}$)
& $0.37 \pm 0.01$
& $2.34 \pm 0.01$
& $12.46 \pm 0.09$
& $79.8 \pm 2.5$
& $1866 \pm 20$ \\

Discrete ($h=2/\sqrt{d}$)
& $0.37 \pm 0.01$
& $2.16 \pm 0.01$
& $12.56 \pm 0.31$
& $78.6 \pm 1.7$
& $1872 \pm 23$ \\
\bottomrule
\end{tabular}
\end{table}

\begin{table}[h!]
\centering
\caption{Sampling runtime on GPU in milliseconds with non-secure  noise (mean $\pm$ standard deviation) over 100 runs.}
\label{tab:sampling_runtime_cuda}
\begin{tabular}{lccccc}
\toprule
& $d=10^3$ & $d=10^4$ & $d=10^5$ & $d=10^6$ & $d=10^7$ \\
\midrule
Dithered ($\xi=0.5\sigma$)
& $0.80 \pm 0.44$
& $1.40 \pm 0.30$
& $1.52 \pm 0.29$
& $2.19 \pm 0.36$
& $3.38 \pm 0.39$ \\

Dithered ($\xi=\sigma$)
& $0.62 \pm 0.30$
& $1.22 \pm 0.35$
& $\mathbf{1.43 \pm 0.15}$
& $1.92 \pm 0.43$
& $\mathbf{3.14 \pm 0.09}$ \\

Dithered ($\xi=2\sigma$)
& $\mathbf{0.54 \pm 0.21}$
& $\mathbf{1.09 \pm 0.45}$
& $1.43 \pm 0.17$
& $\mathbf{1.75 \pm 0.41}$
& $3.16 \pm 0.25$ \\
\midrule
Discrete ($h=0.5/\sqrt{d}$)
& $1.50 \pm 0.03$
& $3.52 \pm 0.05$
& $13.22 \pm 0.44$
& $104.4 \pm 0.8$
& $1635 \pm 7$ \\

Discrete ($h=1/\sqrt{d}$)
& $1.52 \pm 0.39$
& $3.49 \pm 0.06$
& $13.18 \pm 0.14$
& $104.3 \pm 0.6$
& $1639 \pm 5$ \\

Discrete ($h=2/\sqrt{d}$)
& $1.47 \pm 0.04$
& $3.20 \pm 0.13$
& $13.19 \pm 0.10$
& $104.3 \pm 0.8$
& $1636 \pm 5$ \\
\bottomrule
\end{tabular}
\end{table}

\begin{table}[h!]
\centering
\caption{Sampling runtime on CPU in milliseconds with cryptographically secure noise (mean $\pm$ standard deviation) over 100 runs.}
\label{tab:sampling_runtime_cpu_secure}
\begin{tabular}{lccccc}
\toprule
& $d=10^3$ & $d=10^4$ & $d=10^5$ & $d=10^6$ & $d=10^7$ \\
\midrule
Dithered ($\xi=0.5\sigma$)
& $0.28 \pm 0.16$
& $0.72 \pm 0.37$
& $1.39 \pm 0.07$
& $12.66 \pm 0.28$
& $\mathbf{105 \pm 1}$ \\

Dithered ($\xi=\sigma$)
& $0.18 \pm 0.07$
& $0.41 \pm 0.09$
& $1.37 \pm 0.06$
& $12.56 \pm 0.25$
& $126 \pm 21$ \\

Dithered ($\xi=2\sigma$)
& $\mathbf{0.17 \pm 0.06}$
& $\mathbf{0.37 \pm 0.10}$
& $\mathbf{1.36 \pm 0.03}$
& $\mathbf{12.40 \pm 0.41}$
& $125 \pm 18$ \\
\midrule
Discrete ($h=0.5/\sqrt{d}$)
& $0.50 \pm 0.01$
& $3.46 \pm 0.01$
& $27.04 \pm 1.8$
& $175.0 \pm 2.3$
& $3609 \pm 39$ \\

Discrete ($h=1/\sqrt{d}$)
& $0.50 \pm 0.01$
& $3.47 \pm 0.03$
& $26.13 \pm 1.5$
& $175.9 \pm 2.9$
& $3574 \pm 18$ \\

Discrete ($h=2/\sqrt{d}$)
& $0.51 \pm 0.12$
& $3.50 \pm 0.21$
& $25.66 \pm 1.4$
& $176.1 \pm 3.1$
& $3574 \pm 16$ \\
\bottomrule
\end{tabular}
\end{table}

\begin{table}[h!]
\centering
\caption{Sampling runtime on GPU in milliseconds with cryptographically secure noise (mean $\pm$ standard deviation) over 100 runs.}
\label{tab:sampling_runtime_cuda_secure}
\begin{tabular}{lccccc}
\toprule
& $d=10^3$ & $d=10^4$ & $d=10^5$ & $d=10^6$ & $d=10^7$ \\
\midrule
Dithered ($\xi=0.5\sigma$)
& $0.72 \pm 0.90$
& $1.48 \pm 0.20$
& $2.24 \pm 0.28$
& $9.12 \pm 0.42$
& $74.66 \pm 0.55$ \\

Dithered ($\xi=\sigma$)
& $0.60 \pm 0.25$
& $1.38 \pm 0.36$
& $2.26 \pm 0.29$
& $8.98 \pm 0.44$
& $\mathbf{74.43 \pm 0.45}$ \\

Dithered ($\xi=2\sigma$)
& $\mathbf{0.58 \pm 0.23}$
& $\mathbf{1.16 \pm 0.43}$
& $\mathbf{2.19 \pm 0.19}$
& $\mathbf{8.75 \pm 0.43}$
& $74.80 \pm 0.48$ \\
\midrule
Discrete ($h=0.5/\sqrt{d}$)
& $1.76 \pm 0.04$
& $6.64 \pm 0.64$
& $35.79 \pm 2.1$
& $300.6 \pm 5.6$
& $4793 \pm 16$ \\

Discrete ($h=1/\sqrt{d}$)
& $1.75 \pm 0.04$
& $6.37 \pm 0.59$
& $35.39 \pm 0.31$
& $299.4 \pm 6.2$
& $4797 \pm 14$ \\

Discrete ($h=2/\sqrt{d}$)
& $1.76 \pm 0.03$
& $6.24 \pm 0.82$
& $35.58 \pm 0.68$
& $299.1 \pm 5.6$
& $4796 \pm 15$ \\
\bottomrule
\end{tabular}
\end{table}

\begin{table}[h!]
\centering
\caption{CPU sampling runtime with cryptographically secure noise
(mean $\pm$ standard deviation) over 100 runs, in milliseconds. We set
$\xi=\sigma$ for the dithered Gaussian mechanism and $h=1/\sqrt{d}$ for the discrete Gaussian mechanism. All methods are evaluated using 16 cores of an AMD EPYC 9654 CPU.}
\label{tab:sampling_runtime_comparison}
\begin{tabular}{lcccc}
\toprule
& $d=10^3$ & $d=10^4$ & $d=10^5$ & $d=10^6$ \\
\midrule
Dithered Gaussian (ours)
& $\mathbf{0.18 \pm 0.07}$
& $\mathbf{0.41 \pm 0.09}$
& $\mathbf{1.37 \pm 0.06}$
& $\mathbf{12.56 \pm 0.25}$ \\

Discrete Gaussian (jax\_privacy)
& $0.50 \pm 0.01$
& $3.47 \pm 0.03$
& $26.13 \pm 1.5$
& $175.9 \pm 2.9$\\

Discrete Gaussian (OpenDP)
& $30.82 \pm 0.86$
& $342.1 \pm 5.5$
& $2193 \pm 13$
& $17230 \pm 130$ \\

Discrete Gaussian (diffprivlib)
& $24.67 \pm 0.76$
& $267.4 \pm 6.1$
& $1906 \pm 5$
& $25320 \pm 90$ \\
\bottomrule
\end{tabular}
\end{table}

\begin{table}[h!]
\centering
\caption{Sampling runtime comparison with cryptographically secure noise
(mean $\pm$ standard deviation) over 100 runs, in milliseconds, for discrete
inputs with coordinates bounded to $[-100,100]$ and noise multiplier $\sigma=10$.
All measurements are performed using 16 cores of an AMD EPYC 9654 CPU.}
\label{tab:sampling_runtime_discrete}
\begin{tabular}{lcccc}
\toprule
& $d=10^3$ & $d=10^4$ & $d=10^5$ & $d=10^6$ \\
\midrule
Dithered Gaussian (ours)
& $0.36 \pm 0.13$
& $\mathbf{0.46 \pm 0.09}$
& $\mathbf{1.32 \pm 0.03}$
& $\mathbf{11.51 \pm 0.81}$ \\

Discrete Gaussian (jax\_privacy)
& $\mathbf{0.16 \pm 0.002}$
& $1.24 \pm 0.01$
& $12.27 \pm 0.30$
& $132.4 \pm 2.0$ \\

Discrete Gaussian (OpenDP)
& $14.18 \pm 0.45$
& $140.1 \pm 1.1$
& $1400 \pm 13$
& $13970 \pm 14$ \\

Discrete Gaussian (diffprivlib)
& $22.93 \pm 0.71$
& $210.3 \pm 1.9$
& $2113 \pm 7$
& $21160 \pm 111$ \\
\bottomrule
\end{tabular}
\end{table}

\begin{table}[h!]
\centering
\caption{Runtime on CIFAR-10 over $10$ epochs in seconds across privacy levels (mean $\pm$ sample standard deviation) over 3 runs.}
\label{tab:cifar10_runtime_privacy}
\begin{tabular}{lccccc}
\toprule
& $\varepsilon=0.5$ & $\varepsilon=1$ & $\varepsilon=2$ & $\varepsilon=4$ & $\varepsilon=8$ \\
\midrule
Gaussian (Opacus)
& $\mathbf{126.3 \pm 1.5}$
& $\mathbf{128.9 \pm 5.2}$
& $129.2 \pm 4.6$
& $\mathbf{127.4 \pm 3.3}$
& $\mathbf{126.4 \pm 0.9}$ \\

Opacus Secure Mode
& $135.5 \pm 3.4$
& $136.1 \pm 1.1$
& $133.5 \pm 1.5$
& $136.2 \pm 1.5$
& $134.3 \pm 1.1$ \\
\midrule
Dithered ($\xi=0.5\sigma$)
& $141.8 \pm 20.9$
& $130.3 \pm 2.9$
& $133.4 \pm 9.2$
& $128.1 \pm 1.1$
& $128.8 \pm 1.5$ \\

Dithered ($\xi=\sigma$)
& $138.0 \pm 14.2$
& $129.1 \pm 2.0$
& $129.6 \pm 2.2$
& $129.1 \pm 1.6$
& $131.9 \pm 7.9$ \\

Dithered ($\xi=2\sigma$)
& $129.7 \pm 2.1$
& $139.0 \pm 21.3$
& $\mathbf{128.2 \pm 1.5}$
& $130.0 \pm 4.1$
& $149.4 \pm 35.8$ \\
\midrule
Discrete ($h=0.5/\sqrt{d}$)
& $190.8 \pm 2.4$
& $191.1 \pm 2.6$
& $190.7 \pm 1.9$
& $193.6 \pm 6.8$
& $193.3 \pm 2.4$ \\

Discrete ($h=1/\sqrt{d}$)
& $189.2 \pm 0.3$
& $191.4 \pm 1.9$
& $192.7 \pm 4.7$
& $193.2 \pm 4.0$
& $191.0 \pm 1.2$ \\

Discrete ($h=2/\sqrt{d}$)
& $190.4 \pm 2.7$
& $190.1 \pm 1.1$
& $189.1 \pm 1.0$
& $190.4 \pm 1.9$
& $192.8 \pm 2.9$ \\
\bottomrule
\end{tabular}
\end{table}

\begin{table}[h!]
\centering
\caption{Test accuracy on CIFAR-10 in percent across privacy levels (mean $\pm$ sample standard deviation) over 3 runs.}
\label{tab:cifar10_accuracy_privacy}
\begin{tabular}{lccccc}
\toprule
& $\varepsilon=0.5$ & $\varepsilon=1$ & $\varepsilon=2$ & $\varepsilon=4$ & $\varepsilon=8$ \\
\midrule
Gaussian (Opacus)
& $\mathbf{43.7 \pm 0.3}$
& $47.1 \pm 0.6$
& $51.1 \pm 0.8$
& $53.1 \pm 0.7$
& $56.2 \pm 0.3$ \\
\midrule
Dithered ($\xi=0.5\sigma$)
& $42.3 \pm 0.4$
& $\mathbf{47.2 \pm 0.4}$
& $\mathbf{51.5 \pm 1.3}$
& $53.9 \pm 0.5$
& $56.2 \pm 1.1$ \\

Dithered ($\xi=\sigma$)
& $42.6 \pm 0.3$
& $46.4 \pm 0.7$
& $50.5 \pm 0.5$
& $\mathbf{54.2 \pm 0.3}$
& $\mathbf{56.7 \pm 1.2}$ \\

Dithered ($\xi=2\sigma$)
& $42.5 \pm 0.4$
& $45.5 \pm 0.8$
& $49.4 \pm 0.7$
& $52.9 \pm 1.0$
& $55.3 \pm 1.0$ \\
\midrule
Discrete ($h=0.5/\sqrt{d}$)
& $40.3 \pm 0.1$
& $45.2 \pm 0.6$
& $48.5 \pm 0.2$
& $51.0 \pm 0.7$
& $53.6 \pm 0.5$ \\

Discrete ($h=1/\sqrt{d}$)
& $38.8 \pm 1.4$
& $43.0 \pm 0.8$
& $46.5 \pm 0.5$
& $49.8 \pm 0.1$
& $51.1 \pm 1.0$ \\

Discrete ($h=2/\sqrt{d}$)
& $34.9 \pm 0.7$
& $40.1 \pm 0.8$
& $41.3 \pm 1.5$
& $43.5 \pm 0.8$
& $45.8 \pm 1.0$ \\
\bottomrule
\end{tabular}
\end{table}

\begin{table}[h!]
\centering
\caption{Test accuracy of BERT-Base on IMDb in percent across privacy levels
(mean $\pm$ sample standard deviation) over 3 runs.}
\label{tab:imdb_bert_accuracy_privacy}
\begin{tabular}{lccccc}
\toprule
& $\varepsilon=0.5$ & $\varepsilon=1$ & $\varepsilon=2$ & $\varepsilon=4$ & $\varepsilon=8$ \\
\midrule
Gaussian (Opacus)
& $\mathbf{85.6 \pm 0.2}$
& $86.9 \pm 0.3$
& $88.0 \pm 0.4$
& $\mathbf{88.9 \pm 0.2}$
& $\mathbf{89.4 \pm 0.3}$ \\
\midrule
Dithered ($\xi=0.5\sigma$)
& $85.4 \pm 0.3$
& $86.9 \pm 0.1$
& $88.2 \pm 0.1$
& $88.8 \pm 0.2$
& $89.4 \pm 0.2$ \\

Dithered ($\xi=\sigma$)
& $85.1 \pm 0.4$
& $\mathbf{87.0 \pm 0.2}$
& $\mathbf{88.3 \pm 0.2}$
& $88.8 \pm 0.3$
& $89.0 \pm 0.6$ \\

Dithered ($\xi=2\sigma$)
& $85.0 \pm 0.5$
& $86.4 \pm 0.2$
& $87.9 \pm 0.3$
& $88.3 \pm 0.4$
& $89.1 \pm 0.2$ \\
\midrule
Discrete ($h=0.5/\sqrt{d}$)
& $84.9 \pm 0.7$
& $86.3 \pm 0.4$
& $87.3 \pm 0.3$
& $88.4 \pm 0.2$
& $88.8 \pm 0.1$ \\

Discrete ($h=1/\sqrt{d}$)
& $83.9 \pm 0.3$
& $85.7 \pm 0.1$
& $87.3 \pm 0.2$
& $88.0 \pm 0.2$
& $88.4 \pm 0.4$ \\

Discrete ($h=2/\sqrt{d}$)
& $82.0 \pm 0.8$
& $84.8 \pm 0.2$
& $86.2 \pm 0.3$
& $87.3 \pm 0.3$
& $87.9 \pm 0.4$ \\
\bottomrule
\end{tabular}
\end{table}

\begin{table}[h!]
\centering
\caption{Runtime of BERT-Base on IMDb over $10$ epochs in seconds across privacy levels
(mean $\pm$ sample standard deviation) over 3 runs.}
\label{tab:imdb_bert_runtime_privacy}
\begin{tabular}{lccccc}
\toprule
& $\varepsilon=0.5$ & $\varepsilon=1$ & $\varepsilon=2$ & $\varepsilon=4$ & $\varepsilon=8$ \\
\midrule
Gaussian (Opacus)
& $\mathbf{3896 \pm 14}$
& $\mathbf{3896 \pm 13}$
& $\mathbf{3895 \pm 14}$
& $\mathbf{3895 \pm 14}$
& $\mathbf{3895 \pm 14}$ \\
Opacus Secure Mode
& $4762 \pm 10$
& $4757 \pm 6$
& $4766 \pm 9$
& $4755 \pm 4$
& $4752 \pm 8$ \\
\midrule
Dithered ($\xi=0.5\sigma$)
& $4167 \pm 13$
& $4158 \pm 14$
& $4158 \pm 14$
& $4158 \pm 13$
& $4157 \pm 14$ \\

Dithered ($\xi=\sigma$)
& $4152 \pm 14$
& $4153 \pm 15$
& $4153 \pm 15$
& $4152 \pm 14$
& $4154 \pm 14$ \\

Dithered ($\xi=2\sigma$)
& $4149 \pm 14$
& $4149 \pm 13$
& $4149 \pm 13$
& $4148 \pm 12$
& $4149 \pm 14$ \\
\midrule
Discrete ($h=0.5/\sqrt{d}$)
& $12680 \pm 112$
& $12534 \pm 33$
& $12583 \pm 62$
& $12540 \pm 20$
& $12537 \pm 20$ \\

Discrete ($h=1/\sqrt{d}$)
& $12556 \pm 14$
& $12594 \pm 81$
& $12517 \pm 8$
& $12606 \pm 69$
& $12544 \pm 33$ \\

Discrete ($h=2/\sqrt{d}$)
& $12608 \pm 75$
& $12600 \pm 55$
& $12560 \pm 36$
& $12557 \pm 44$
& $12545 \pm 62$ \\
\bottomrule
\end{tabular}
\end{table}

\clearpage
\section{Lower Bound}\label{app:lower_bound}

Canonne, Su, and Vadhan study the randomness complexity of differentially
private mechanisms for the summation problem, where the input \(X\) is a
dataset of \(n\) vectors in \([0,1]^d\). In this setting, the function has
\(\ell_1\)-sensitivity \(\Delta_1(f)=d\). They prove a lower bound on the
amount of randomness required by accurate differentially private mechanisms
for this problem \cite[Cor.~4.2]{canonne2025randomness}.
Specifically, for an $(\varepsilon,\delta)$-differentially private mechanism with Hamming adjacency and parameters $\varepsilon\le 1/d$ and $\delta\le 1/(6d^2)$, any mechanism that retains nontrivial accuracy (say $\beta \le 1/d$ and $\alpha \le n/2 - 1$ on inputs of length $n$) has expected randomness complexity
\begin{equation}
  \E[H(\M(X))]\ \ge \log_2 d\ -\ O(1).
\label{eq:csv-lower}
\end{equation}
This holds even under \emph{bounded differential privacy} where the data size $n$ is public information.

In this section, we show an improved lower bound for approximate
differential privacy by separating the public randomness \(U\) from the
private randomness needed to ensure differential privacy. We consider the
following setting. Given a clipping radius \(C\in\mathbb{N}\), the function we consider is the identity function \(f(z)=z\) defined on \(\ZZ\), where integers \(z\) and \(z'\) are neighboring if and only if
\(|z-z'| \leq C\), so that the \(\ell_1\)-sensitivity is \(C\).
A mechanism for \(f\) implies a mechanism for summation of integers in \(\{0,\dots,C\}\).
We compute the private entropy of a mechanism \(\M\) defined on \(\ZZ\) as
\[
H(\M)\ :=\ \sup_{z\in \ZZ}\ \E_U\!\big[\,H(\M_U(f(z)))\,\big],
\]
where the expectation is over public randomness $U$ and $H(\M_U(f(z)))$ is the Shannon entropy of the output distribution of $\M$ on input $f(z)$ given $U$.

\mainlower*

The proof proceeds by fixing a suitable value of the public randomness and using a graph-theoretic argument based on component stability and barrier structure to derive a lower bound on the achievable accuracy.

\begin{proof}
Write \(H=H(\M)\).
We also use the standard group-privacy consequence of \((\varepsilon,\delta)\)-differential privacy: if two inputs are connected by \(t\) neighboring steps, then for every event \(E\),
\begin{equation}\label{eq:group-privacy-lower}
\Pr[\M_U(z')\in E]
\ge
e^{-t\varepsilon}\left(\Pr[\M_U(z)\in E]
-\delta\sum_{j=0}^{t-1}e^{j\varepsilon}\right)
\end{equation}
for every fixed value \(U\) of the public randomness.

We first record a simple consequence of accuracy.
For a sufficiently small absolute constant \(a>0\), no mechanism satisfying the hypotheses can have \(\alpha<aC/\varepsilon\).
Indeed, let \(z'=z+2\alpha+1\), so the intervals
\[
I_z=[z-\alpha,z+\alpha],
\qquad
I_{z'}=[z'-\alpha,z'+\alpha]
\]
are disjoint.
Accuracy gives \(\E_U\Pr[\M_U(z)\in I_z]\ge 1-\beta\) and \(\E_U\Pr[\M_U(z')\in I_{z'}]\ge 1-\beta\).
Since \(I_z\cap I_{z'}=\emptyset\), the second inequality implies \(\E_U\Pr[\M_U(z')\in I_z]\le \beta\).
On the other hand, averaging \eqref{eq:group-privacy-lower} over \(U\) and using \(\delta\le c_0\varepsilon\), with \(a\) and \(c_0\) small enough, gives a constant lower bound larger than \(\beta\) on \(\E_U\Pr[\M_U(z')\in I_z]\) whenever \(\alpha<aC/\varepsilon\), a contradiction.
Thus, in the non-vacuous case we may assume
\begin{equation}\label{eq:alpha-at-least-privacy-scale}
\alpha\ge aC/\varepsilon.
\end{equation}

Fix \(n\) and write \(Q_n=\{1,\ldots,n\}\).
Since the mechanism satisfies the stated error guarantee on every input,
\[
\E_U\Pr_{z\sim Q_n}\!\left[|\M_U(z)-z|>\alpha\right]\le \beta.
\]
Also, by the definition of \(H\),
\[
\E_U\E_{z\sim Q_n}\!\left[H(\M_U(z))\right]\le H.
\]
Markov's inequality therefore gives a value \(U^*\) such that
\begin{align}
\Pr_{z\sim Q_n}\!\left[|\M_{U^*}(z)-z|>\alpha\right]&\le 2\beta, \label{eq:fixed-u-accuracy}\\
\E_{z\sim Q_n}\!\left[H(\M_{U^*}(z))\right]&\le 3H. \label{eq:fixed-u-entropy}
\end{align}
Let
\[
S=\{z\in Q_n: H(\M_{U^*}(z))\ge 1/2\}.
\]
By \eqref{eq:fixed-u-entropy}, \(|S|\le 6Hn\).
For every \(z\notin S\), the output distribution of \(\M_{U^*}(z)\) has an atom \(m_z\) of probability at least \(7/8\): otherwise its Shannon entropy would be at least the binary entropy \(h_2(1/8)>1/2\).

Choose a sufficiently small absolute constant \(b>0\), and set
\[
k=\max\{1,\lfloor bC/\varepsilon\rfloor\}.
\]
If \(z,z'\notin S\) and \(|z-z'|\le k\), then \(m_z=m_{z'}\).
To see this, let \(t=\lceil |z-z'|/C\rceil\).
By the choice of \(k\), \(t\varepsilon\le b+\varepsilon\), and the additive term in \eqref{eq:group-privacy-lower} is at most an absolute constant that can be made smaller than \(1/4\) by choosing \(c_0\) small enough.
Applying \eqref{eq:group-privacy-lower} to the event \(\{\M_{U^*}(\cdot)=m_z\}\) gives probability greater than \(1/8\) under input \(z'\).
Since \(z'\notin S\), all atoms other than its atom of mass at least \(7/8\) have probability at most \(1/8\), so \(m_z=m_{z'}\).

Build a graph \(G\) whose vertices are \(Q_n\setminus S\), connecting two vertices when their distance is at most \(k\).
Each connected component is an interval in the line \(Q_n\), and consecutive components are separated by a run of at least \(k\) points of \(S\).
Moreover, by the stability just proved, all vertices in a component \(D\) share a common modal output \(m_D\).

For a component \(D\), at most \(2\alpha+1\) of its vertices satisfy \(|z-m_D|\le\alpha\).
Thus, if \(|D|>4\alpha\), then for uniformly random \(z\in D\),
\[
\Pr\!\left[|\M_{U^*}(z)-z|>\alpha\right]
\ge
\frac78\left(1-\frac{2\alpha+1}{|D|}\right)
\ge \frac{7}{20}.
\]
Together with \eqref{eq:fixed-u-accuracy}, this implies that the total number of vertices lying in components of size larger than \(4\alpha\) is at most \(6\beta n\).
Consequently the total size of components of size at most \(4\alpha\) is at least
\[
(1-6\beta)n-|S|.
\]

On the other hand, the barrier property bounds the number of components by \(|S|/k+1\).
Hence the total size of components of size at most \(4\alpha\) is at most
\[
4\alpha\left(\frac{|S|}{k}+1\right)
\le
4\alpha\left(\frac{6Hn}{k}+1\right).
\]
Combining the last two displays and dividing by \(n\), then letting \(n\to\infty\), gives
\[
1-6\beta
\le
6H+\frac{24\alpha H}{k}.
\]
By \eqref{eq:alpha-at-least-privacy-scale} and the definition of \(k\), the ratio \(\alpha/k\) is bounded below by an absolute constant.
Since \(\beta<1/10\), this implies
\[
H\ge c'\frac{k}{\alpha}
\ge c\,\frac{C}{\varepsilon\alpha}
\]
for absolute constants \(c',c>0\).
This is the desired lower bound.
\end{proof}

\section{Properties of the Dithered Gaussian Mechanism}\label{app:proofs}

This section contains full proofs of properties of the dithered Gaussian mechanism that were sketched in the main body.

\ditheredGaussianRoundingLemma*
\label{lem:uniform_distribution_proof}
\begin{proof}

The dithered Gaussian mechanism is defined in~\eqref{eq:mechanism}:

\begin{equation}
    \M(f(X))_i
=
\xi\left(
\left\lfloor
\frac{f(X)_i+y_i}{\xi}-\gamma_i+\frac12
\right\rfloor
+\gamma_i
\right),
\end{equation}
where $\gamma_i=(ia+b)\bmod 1$ with $a,b\sim \mathrm{Uniform}([0,1)^2)$ independent of $X$ and $y$.
Consider the value

\begin{equation}
    \frac{u_i}{\xi}
    :=
    \left\lfloor t_i - \gamma_i + \frac{1}{2} \right\rfloor
    -t_i+\gamma_i,
    \qquad
    \text{where}
    \qquad
    t_i = \frac{f(X)_i+y_i}{\xi}.
\end{equation}

We will show that $\frac{u_i}{\xi}$ is uniformly distributed on $[-1/2,1/2]$ independently of $t_i$, which is sufficient for the proof of the lemma.
Let $\{\cdot\}$ denote the fractional part.
Then

\begin{equation}
    t_i-\gamma_i+\frac{1}{2}
    =
    \left\lfloor t_i-\gamma_i+\frac{1}{2} \right\rfloor
    +
    \left\{t_i-\gamma_i+\frac{1}{2}\right\}.
\end{equation}

Therefore,

\begin{equation}
    \frac{u_i}{\xi}
    =
    \frac{1}{2}
    -
    \left\{t_i-\gamma_i+\frac{1}{2}\right\}.
\end{equation}

It remains to show that the fractional part is uniformly distributed on $[0,1)$ and independent of $t_i$.
Fix arbitrary values of $t_i$ and $a$.
Conditional on these values, the quantity $t_i-ia+1/2$ is a constant, while $b$ remains uniform on $[0,1)$.
Hence, for any measurable set $\mathcal{U}\subseteq[0,1)$,

\begin{equation}
    \Pr\left(
    \left\{t_i-ia+\frac{1}{2}-b\right\}\in \mathcal{U}
    \,\middle|\, t_i,a
    \right)
    =
    \lambda(\mathcal{U}),
\end{equation}

where $\lambda$ denotes Lebesgue measure on $[0,1)$.
This follows because the map

\begin{equation}
    b
    \mapsto
    \left\{t_i-ia+\frac{1}{2}-b\right\}
\end{equation}

is a measure-preserving shift and reflection modulo $1$.
Therefore,

\begin{equation}
    \left\{t_i-\gamma_i+\frac{1}{2}\right\}
    \sim \mathrm{Uniform}([0,1))
\end{equation}

independently of $t_i$.
\end{proof}

\ditheredGaussianTVBound*

\begin{proof}

The proof is based on Lemma~2 of \citet{christofides2009bounds}, together with an explicit computation of $\|\phi''\|_1$ and substitution of the variance of the uniform random variable. For completeness, we provide the full proof here.

By scaling all random variables by \(1/\sigma\), it suffices to compare the standard Gaussian \(G\sim\mathcal N(0,1)\) to \(G+U_r\), where \(U_r\sim\mathrm{Uniform}(-r/2,r/2)\) and \(r=\xi/\sigma\).
The density of \(G+U_r\) is
\begin{equation*}
    q_r(x)
    =
    \frac{1}{r}
    \int_{x-r/2}^{x+r/2}\phi(t)\,dt
    =
    \frac{\Phi(x+r/2)-\Phi(x-r/2)}{r}.
\end{equation*}
Thus
\begin{equation*}
    \dTV(G_\sigma+U_\xi,G_\sigma)
    =
    \frac12
    \int_{-\infty}^{\infty}
    \left|
    \frac{\Phi(x+r/2)-\Phi(x-r/2)}{r}
    -
    \phi(x)
    \right|\,dx ,
\end{equation*}
since total variation distance between two distributions with densities is one half of the \(L^1\)-distance between the densities.

We next prove the bound.
Using \(\mathbb E[U_r]=0\), we can write
\begin{equation*}
    q_r(x)-\phi(x)
    =
    \mathbb E\!\left[
        \phi(x+U_r)-\phi(x)-U_r\phi'(x)
    \right].
\end{equation*}
For a fixed \(v\in[-r/2,r/2]\), Taylor's formula with integral remainder gives
\begin{equation*}
    \phi(x+v)-\phi(x)-v\phi'(x)
    =
    \int_0^v (v-s)\phi''(x+s)\,ds.
\end{equation*}
Taking \(L^1\)-norms and using translation-invariance of the Lebesgue integral,
\begin{align}
    \left\|\phi(\cdot+v)-\phi-v\phi'\right\|_1
    &\le
    \frac{v^2}{2}\,\|\phi''\|_1.
\end{align}
Therefore
\begin{equation*}
    \|q_r-\phi\|_1
    \le
    \frac{\mathbb E[U_r^2]}{2}\,\|\phi''\|_1
    =
    \frac{r^2}{24}\,\|\phi''\|_1.
\end{equation*}
Since \(\phi''(x)=(x^2-1)\phi(x)\), we have
\begin{align}
    \|\phi''\|_1
    &=
    \int_{-\infty}^{\infty}|x^2-1|\phi(x)\,dx \notag\\
    &=
    2\int_{-1}^{1}(1-x^2)\phi(x)\,dx
    =
    4\phi(1),
\end{align}
where the last equality follows from \((x\phi(x))'=(1-x^2)\phi(x)\).
Thus
\begin{equation*}
    \dTV(G+U_r,G)
    =
    \frac12\|q_r-\phi\|_1
    \le
    \frac{\phi(1)}{12}r^2,
\end{equation*}
and \(\phi(1)/12<0.0202\), which gives~\eqref{eq:dithered_gaussian_tv_bound}.
\end{proof}
We remark that the constant \(\phi(1)/12\) is sharp for a bound of this form.

\medskip

Before proving Lemma~\ref{lem:dithered_gaussian_entropy_bound}, we show an auxiliary upper bound on the entropy of an integer-valued random variable
in terms of its variance.

\begin{lemma}[Discrete entropy bound {\citep[Theorem 9.7.1]{cover1999elements}}]
\label{lem:entropy_bound}
Let $Z$ be an integer-valued random variable with finite variance.
Then its Shannon binary entropy satisfies
\begin{equation}
    H(Z)\le \frac12 \log_2\left(2\pi e\left(\mathrm{Var}(Z)+\frac1{12}\right)\right).
\end{equation}
\end{lemma}

\begin{proof}
We provide a more detailed, self-contained proof.
Let $U\sim \mathrm{Uniform}(-1/2,1/2]$ be independent of $Z$.
Then $Z+U$ has density
\begin{equation}
    \rho(x):=\rho_{Z+U}(x)
    =
    \sum_{k=-\infty}^{+\infty}
    \Pr[Z=k]\mathbf{1}_{\{k-\frac12 < x \le k+\frac12\}}.
\end{equation}
The density $\rho$ is nonnegative and pointwise bounded by $\max_k\Pr[Z=k]\le 1$.
Therefore, if $\mu$ denotes the distribution of $Z+U$, then for every measurable set $\mathcal{A}$,
\begin{equation}
    \mu(\mathcal{A})
    =
    \int_{\mathcal{A}}\rho(x)\,dx
    \le
    \lambda(\mathcal{A}).
\end{equation}
In particular, if $\lambda(\mathcal{A})=0$, then $\mu(\mathcal{A})=0$.
Hence $\mu$ is absolutely continuous with respect to the Lebesgue measure $\lambda$.
Therefore, we can compute the differential entropy of $Z+U$:

\begin{align}
    h(Z+U)
    &= -\int_{-\infty}^{+\infty}\rho(x)\log_2\rho(x)\,dx = \sum_{k=-\infty}^{+\infty}
    -\int_{k-\frac12}^{k+\frac12} \rho(x)\log_2\rho(x)\,dx\\
    &= \sum_{k=-\infty}^{+\infty}
    -\int_{k-\frac12}^{k+\frac12}
    \Pr[Z=k]\log_2\Pr[Z=k]\,dx \\
    &= -\sum_{k=-\infty}^{+\infty}\Pr[Z=k]\log_2\Pr[Z=k] = H(Z).
\end{align}

Thus, the differential entropy of $Z+U$ coincides with the discrete entropy of $Z$.
Since the Gaussian distribution maximizes differential entropy among all absolutely continuous random variables with fixed variance,
\begin{equation}
    h(Z+U)\le \frac12\log_2\left(2\pi e\,\mathrm{Var}(Z+U)\right).
\end{equation}

Using independence of $Z$ and $U$,
\begin{equation}
    \mathrm{Var}(Z+U)
    = \mathrm{Var}(Z)+\mathrm{Var}(U)
    = \mathrm{Var}(Z)+\frac1{12}.
\end{equation}

Combining the above identities yields
\begin{equation}
    H(Z)\le \frac12 \log_2\left(2\pi e\left(\mathrm{Var}(Z)+\frac1{12}\right)\right),
\end{equation}
which completes the proof.
\end{proof}

\begin{figure}[t!]
    \begin{subfigure}{0.45\linewidth}
        \centering
        \includegraphics[width=\linewidth]{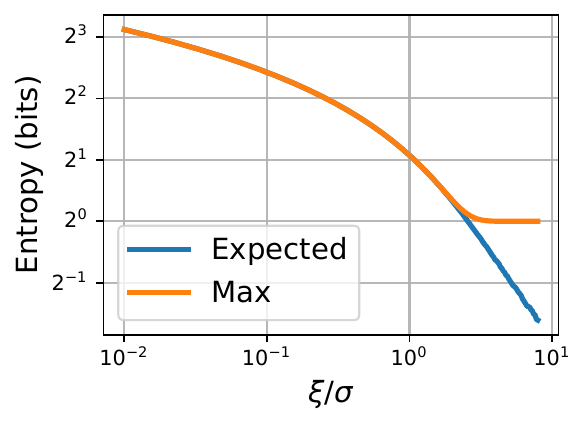}
        \caption{Entropy of the dithered Gaussian distribution.}
        \label{fig:entropy_vs_xi_sigma}
    \end{subfigure}
    \hfill
    \begin{subfigure}{0.45\linewidth}
        \centering
        \includegraphics[width=\linewidth]{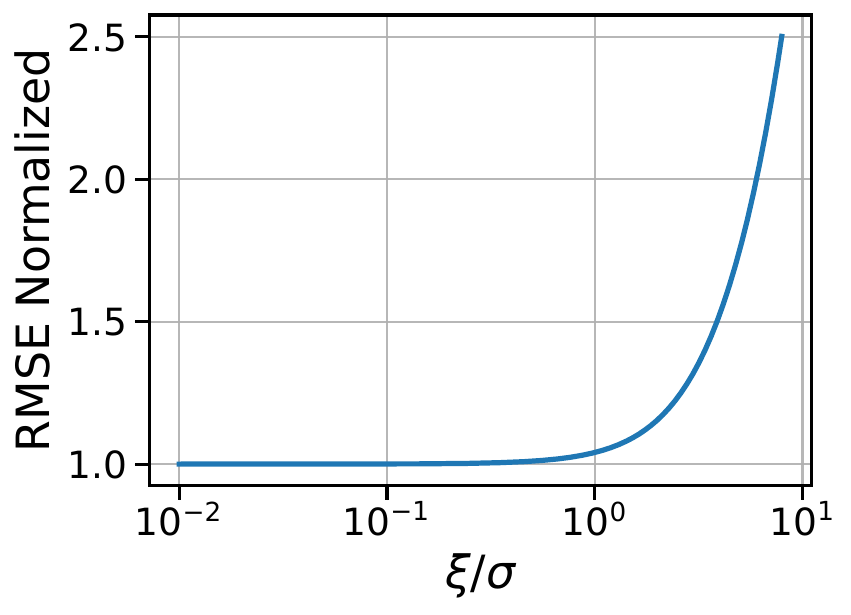}
        \caption{Root mean squared error normalized by $\sigma$.}
        \label{fig:mse_vs_xi_sigma}
    \end{subfigure}

    \caption{The expected and maximum entropy (a) and normalized RMSE of the dithered Gaussian mechanism (b) as functions of $\xi / \sigma$. }
    \label{fig:entropy_mse_vs_xi_sigma}
\end{figure}

\ditheredGaussianEntropyBound*
\label{lem:dithered_gaussian_entropy_bound_proof}

\begin{proof}
Recall from \eqref{eq:Z_k} that, conditionally on $\gamma_i$ and $f(X)_i$, the distribution of $Z_i$ is given by
\begin{equation}
\Pr[Z_i=k]
=
\Phi\!\left(
\frac{\xi}{\sigma}\left(k+\gamma_i+\frac12-\frac{f(X)_i}{\xi}\right)
\right)
-
\Phi\!\left(
\frac{\xi}{\sigma}\left(k+\gamma_i-\frac12-\frac{f(X)_i}{\xi}\right)
\right).
\end{equation}
Given the grid $\gamma$, the random variables $Z_i$ are independent.
Therefore,
\begin{equation}
\label{eq:entropy_Z_bound}
    H(Z\mid \gamma)=\sum_{i=1}^{d} H(Z_i\mid \gamma).
\end{equation}
By Lemma~\ref{lem:entropy_bound}, for each coordinate we have
\begin{equation}
    H(Z_i\mid \gamma)
    \le
    \frac12 \log_2\left(
        2\pi e\left(\mathrm{Var}(Z_i\mid \gamma)+\frac1{12}\right)
    \right).
\end{equation}
Now consider the random variable
$Y_i\sim \mathcal{N}\left(\frac{f(X)_i}{\xi}-\gamma_i,\frac{\sigma^2}{\xi^2}\right)$.
Then $Z_i$ can be obtained by rounding $Y_i$ to the nearest integer, since
\begin{equation}
    \Pr\left(k-\frac12 \le Y_i <  k+\frac12\right)
    =
    \Pr[Z_i=k].
\end{equation}
Using the Cauchy-Schwarz inequality, we get
\begin{equation}
    \mathrm{Var}(Z_i\mid \gamma_i)
    =
    \mathrm{Var}\left(Y_i+(Z_i-Y_i)\mid \gamma_i\right)
    \le
    \left(
        \sqrt{\mathrm{Var}(Y_i\mid \gamma_i)}
        +
        \sqrt{\mathrm{Var}(Z_i-Y_i\mid \gamma_i)}
    \right)^2.
\end{equation}
Since $\mathrm{Var}(Y_i\mid \gamma_i)=\sigma^2/\xi^2$ and
$|Z_i-Y_i|\le 1/2$, Popoviciu's inequality \citep{popoviciu1935equations} gives
\begin{equation}
    \mathrm{Var}(Z_i-Y_i\mid \gamma_i)\le \frac14.
\end{equation}
Hence
\begin{equation}
\label{eq:dithered_gaussian_variance_upper_bound}
    \mathrm{Var}(Z_i\mid \gamma_i)
    \le
    \left(\frac{\sigma}{\xi}+\frac12\right)^2.
\end{equation}
Combining the above inequalities, we obtain
\begin{equation}
    H(Z\mid \gamma)
    \le
    \frac{d}{2}\log_2\left(
        2\pi e\left[
            \left(\frac{\sigma}{\xi}+\frac12\right)^2+\frac1{12}
        \right]
    \right).
\end{equation}
The right-hand side does not depend on $\gamma$, so taking expectation over $\gamma$ does not change the bound.

\paragraph{Conditional Entropy Bound.}
In the regime $\sigma/\xi \to 0$, the Gaussian noise $Y_i$ affects the rounded value only with small probability. For each coordinate $i$, let
$t_{x,\gamma}^{i}
:=\left\lfloor \frac{f(X)_i}{\xi}-\gamma_i+\frac{1}{2}\right\rfloor$
denote the rounded value in the absence of Gaussian noise, and let
$r_{x,\gamma}^{i}
:=\left\{\frac{f(X)_i}{\xi}-\gamma_i+\frac{1}{2}\right\}$
denote the corresponding fractional part. We define
$p_{\gamma}^{i}
:=\mathbb{P}\left(Z_i=t_{x,\gamma}^{i}\mid\gamma\right)$
as the probability that the Gaussian perturbation does not change the rounded value. This probability is given explicitly by
\begin{equation}
\begin{aligned}
    p_{\gamma}^{i}
    &=
    \Phi\left(
        \frac{\xi\left(t_{x,\gamma}^{i}+\gamma_i+\frac{1}{2}\right)
        -f(X)_i}{\sigma}
    \right)
    -
    \Phi\left(
        \frac{\xi\left(t_{x,\gamma}^{i}+\gamma_i-\frac{1}{2}\right)
        -f(X)_i}{\sigma}
    \right) \\
    &=
    \Phi\left(\frac{\xi\left(1-r_{x,\gamma}^{i}\right)}{\sigma}\right)
    -
    \Phi\left(-\frac{\xi r_{x,\gamma}^{i}}{\sigma}\right),
\end{aligned}
\end{equation}
where the second equality follows from
$t_{x,\gamma}^{i}+r_{x,\gamma}^{i}
=\frac{f(X)_i}{\xi}-\gamma_i+\frac{1}{2}$.

Since $\gamma_i$ is uniformly distributed on $[0,1]$,
$r_{x,\gamma}^{i}$ is also uniformly distributed on $[0,1]$. Therefore,
\begin{equation}
\begin{aligned}
    \mathbb{E}_{\gamma}\left[p_{\gamma}^{i}\right]
    &=
    \int_0^1
    \left[
        \Phi\left(\frac{\xi(1-r)}{\sigma}\right)
        -
        \Phi\left(-\frac{\xi r}{\sigma}\right)
    \right]dr \\
    &=
    \frac{\sigma}{\xi}
    \int_0^{\xi/\sigma}\left(2\Phi(r)-1\right)dr \\
    &=
    2\Phi\left(\frac{\xi}{\sigma}\right)-1
    +\frac{2\sigma}{\xi}\phi\left(\frac{\xi}{\sigma}\right)
    -\sqrt{\frac{2}{\pi}}\frac{\sigma}{\xi}.
\end{aligned}
\end{equation}
Since $\phi(\xi/\sigma)$ and
$2\Phi(\xi/\sigma)-2$ decay exponentially as $\sigma/\xi\to 0$, we have
\begin{equation}
\label{eq:p_gamma_i_asymptotics}
    \mathbb{E}_{\gamma}\left[p_{\gamma}^{i}\right]
    =
    1-\sqrt{\frac{2}{\pi}}\frac{\sigma}{\xi}
    +o\left(\frac{\sigma}{\xi}\right).
\end{equation}

We next separate the event $Z_i=t_{x,\gamma}^{i}$ in the definition of
the entropy. For $k\neq t_{x,\gamma}^{i}$, define the conditional
probabilities
\begin{equation}
    q_{k,\gamma}^{i}
    :=
    \mathbb{P}\left(
        Z_i=k
        \mid
        \gamma,\,
        Z_i\neq t_{x,\gamma}^{i}
    \right)
    =
    \frac{\mathbb{P}(Z_i=k\mid\gamma)}
    {1-p_{\gamma}^{i}}.
\end{equation}
Then
\begin{equation}
\begin{aligned}
    H(Z_i\mid\gamma)
    &=
    -p_{\gamma}^{i}\log_2 p_{\gamma}^{i}
    -
    \sum_{k\neq t_{x,\gamma}^{i}}
    \mathbb{P}(Z_i=k\mid\gamma)
    \log_2\mathbb{P}(Z_i=k\mid\gamma) \\
    &=
    -p_{\gamma}^{i}\log_2 p_{\gamma}^{i}
    -
    (1-p_{\gamma}^{i})
    \sum_{k\neq t_{x,\gamma}^{i}}
    q_{k,\gamma}^{i}
    \log_2\left(
        (1-p_{\gamma}^{i})q_{k,\gamma}^{i}
    \right) \\
    &=
    -p_{\gamma}^{i}\log_2 p_{\gamma}^{i}
    -
    (1-p_{\gamma}^{i})\log_2(1-p_{\gamma}^{i}) -
    (1-p_{\gamma}^{i})
    \sum_{k\neq t_{x,\gamma}^{i}}
    q_{k,\gamma}^{i}\log_2 q_{k,\gamma}^{i} \\
    &=
    h_2\left(p_{\gamma}^{i}\right)
    +
    (1-p_{\gamma}^{i})
    H\left(
        Z_i
        \mid
        \gamma,\,
        Z_i\neq t_{x,\gamma}^{i}
    \right),
\end{aligned}
\end{equation}
where
$h_2(p):=-p\log_2 p-(1-p)\log_2(1-p)$
is the binary entropy function.

By Lemma~\ref{lem:entropy_bound},
\begin{equation}
\begin{aligned}
    &H\left(
        Z_i
        \mid
        \gamma,\,
        Z_i\neq t_{x,\gamma}^{i}
    \right) \leq
    \frac{1}{2}\log_2\left(
        2\pi e\left[
            \operatorname{Var}\left(
                Z_i
                \mid
                \gamma,\,
                Z_i\neq t_{x,\gamma}^{i}
            \right)
            +\frac{1}{12}
        \right]
    \right).
\end{aligned}
\end{equation}

By the conditional law of total variance, \begin{equation} \begin{aligned} \operatorname{Var}(Z_i\mid\gamma) ={}& \left(1-p_{\gamma}^{i}\right) \operatorname{Var}\left( Z_i\mid\gamma,\, Z_i\neq t_{x,\gamma}^{i} \right) \\ &+ p_{\gamma}^{i} \operatorname{Var}\left( Z_i\mid\gamma,\, Z_i=t_{x,\gamma}^{i} \right) \\ &+ p_{\gamma}^{i}\left(1-p_{\gamma}^{i}\right) \left( \mathbb{E}\left[ Z_i\mid\gamma,\, Z_i=t_{x,\gamma}^{i} \right] - \mathbb{E}\left[ Z_i\mid\gamma,\, Z_i\neq t_{x,\gamma}^{i} \right] \right)^2. \end{aligned} \end{equation}
Therefore,
\begin{equation}
    \operatorname{Var}\left(
        Z_i
        \mid
        \gamma,\,
        Z_i\neq t_{x,\gamma}^{i}
    \right)
    \leq
    \frac{\operatorname{Var}(Z_i\mid\gamma)}
    {1-p_{\gamma}^{i}}
    \leq
    \frac{1}{1-p_{\gamma}^{i}}
    \left(
        \frac{\sigma}{\xi}+\frac{1}{2}
    \right)^2,
\end{equation}

where the last inequality follows from
\eqref{eq:dithered_gaussian_variance_upper_bound}. Consequently,
\begin{equation}
\label{eq:coordinate_entropy_bound}
\begin{aligned}
    H(Z_i\mid\gamma)
    \leq\;&
    h_2\left(p_{\gamma}^{i}\right) +
    \frac{1-p_{\gamma}^{i}}{2}
    \log_2\left(
        2\pi e\left[
            \frac{
                \left(
                    \frac{\sigma}{\xi}+\frac{1}{2}
                \right)^2
            }{
                1-p_{\gamma}^{i}
            }
            +\frac{1}{12}
        \right]
    \right).
\end{aligned}
\end{equation}

To take the expectation over $\gamma$, we use Jensen's inequality. The
binary entropy function is concave because
\begin{equation}
    h_2''(x)
    =
    -\frac{1}{x(1-x)\ln 2}
    <0.
\end{equation}
Moreover, for any $a,b>0$, the function
\begin{equation}
    g(x)
    :=
    (1-x)\log_2\left(\frac{a}{1-x}+b\right),
    \qquad
    g''(x)
    =
    -\frac{a^2}
    {(1-x)\left(a+b(1-x)\right)^2\ln 2}
    <0,
\end{equation}
is also concave.
Therefore, the right-hand side of
\eqref{eq:coordinate_entropy_bound} is concave in $p_{\gamma}^{i}$.
Jensen's inequality gives
\begin{equation}
\begin{aligned}
    \mathbb{E}_{\gamma}\left[H(Z_i\mid\gamma)\right]
    \leq\;&
    h_2\left(\mathbb{E}_{\gamma}\left[p_{\gamma}^{i}\right]\right) +
    \frac{
        1-\mathbb{E}_{\gamma}\left[p_{\gamma}^{i}\right]
    }{2}
    \log_2\left(
        2\pi e\left[
            \frac{
                \left(\frac{\sigma}{\xi}+\frac{1}{2}\right)^2
            }{
                1-\mathbb{E}_{\gamma}\left[p_{\gamma}^{i}\right]
            }
            +\frac{1}{12}
        \right]
    \right).
\end{aligned}
\end{equation}

By \eqref{eq:p_gamma_i_asymptotics},
$1-\mathbb{E}_{\gamma}[p_{\gamma}^{i}]
=\Theta(\sigma/\xi)$. Consequently,
$h_2(\mathbb{E}_{\gamma}[p_{\gamma}^{i}])
=O\left(\frac{\sigma}{\xi}\log\frac{\xi}{\sigma}\right)$.
Since
$\left(\frac{\sigma}{\xi}+\frac{1}{2}\right)^2=\Theta(1)$,
the second term has the same asymptotic order. Hence,
\begin{equation}
    \mathbb{E}_{\gamma}\left[H(Z_i\mid\gamma)\right]
    =
    O\left(\frac{\sigma}{\xi}\log\frac{\xi}{\sigma}\right).
\end{equation}

Finally, by subadditivity of entropy,
$H(Z\mid\gamma)\leq\sum_{i=1}^{d}H(Z_i\mid\gamma)$. Taking expectations
and summing over $i$ yields
\begin{equation}
    \mathbb{E}_{\gamma}\left[H(Z\mid\gamma)\right]
    =
    O\left(\frac{d\sigma}{\xi}\log\frac{\xi}{\sigma}\right),
\end{equation}
which concludes the proof.
\end{proof}

\section{Dithered Laplace Mechanism}
\label{sec:dithered_laplace}

The proposed scheme can be generalized to other types of additive-noise mechanisms.
In this section, we discuss the implications of our dithering scheme for the Laplace mechanism.
In a recent paper~ \citep{canonne2025randomness}, Canonne, Su, and Vadhan showed that the randomness complexity of differentially private mechanisms for vector summation can be drastically reduced.
When the statistic $f$ has $\ell_1$-sensitivity $\Delta_1(f)$, there exists an $\varepsilon$-differentially private mechanism that for $\alpha = \Omega(\Delta_1(f)\log(d)/\varepsilon)$ is \((\alpha,\text{poly}(1/d))\)-accurate with expected randomness complexity, as stated in the following corollary.

\begin{restatable}{corollary}{CSVGeneralAlpha}\label{cor:csv-general-alpha}
Let \(\varepsilon>0\), \(\beta\in(0,1)\), and suppose the accuracy target
\(\alpha\) satisfies
\[
\alpha \ge C_0\,\frac{\Delta_1(f)}{\varepsilon}\,\log(2d/\beta)
\]
for a sufficiently large universal constant~\(C_0\).
Then there exists an
\(\varepsilon\)-differentially private mechanism \(\M\) that is
\((\alpha,\beta)\)-accurate and whose expected randomness complexity is
\[
R(\M)
=
O\!\left(
\frac{d\,\Delta_1(f)}{\alpha\,\varepsilon}
\cdot \log(d) \cdot
\log\!\frac{\alpha\,\varepsilon}{\Delta_1(f)\log d}
+ \log d
\right).
\]
\end{restatable}

Choosing error parameter \(\alpha = d^2\log^2(d) / \eps\) minimizes randomness complexity to \(O(\log d)\) bits.
At the other end of the trade-off, error \(\alpha = O(\Delta_1(f) \log(d)/ \eps)\) is asymptotically optimal~\citep{hardt2010geometry,geng2016optimal} among \(\eps\)-differentially private algorithms for \( \beta = 1/\text{poly}(d) \) and results in randomness complexity \(\E[H(\M(f(X)))] = O(d)\) bits.
While the mechanism of~\citet{canonne2025randomness} is not computationally efficient, recent work by~\citet{ghentiyala2026} provides a computationally efficient alternative (with slightly higher randomness complexity). 
By replacing the cumulative distribution function in the dithered Gaussian algorithm with that of the Laplace distribution with parameter $\lambda$ (which controls the privacy level), we obtain the dithered Laplace mechanism; see Algorithm~\ref{alg:dithered-laplace}.
On the upper bound side we show an improvement over the bound of~\citet{canonne2025randomness} in the following theorem.
\begin{algorithm}[t]
\caption{Dithered Laplace Mechanism}
\label{alg:dithered-laplace}
\begin{algorithmic}[1]
\Require Function value $f(X) \in \mathbb{R}^d$, parameters $\xi > 0$, $\lambda > 0$
\State Sample public randomness $(a,b) \sim \mathrm{Uniform}([0,1)^2)$
\For{$i = 1, \dots, d$}
    \State $\gamma_i \gets (a \cdot i + b) \bmod 1$
    \State Sample an integer $Z_i \in \mathbb{Z}$ such that, for every $k \in \mathbb{Z}$,
    \[
    \Pr[Z_i = k]
    =
    F_{\lambda}\!\left(
    \xi\left(k+\gamma_i+\frac12\right)-f(X)_i
    \right)
    -
    F_{\lambda}\!\left(
    \xi\left(k+\gamma_i-\frac12\right)-f(X)_i
    \right),
    \]
    \qquad where $F_{\lambda}(t) = \begin{cases}
\frac{1}{2} e^{t/\lambda}, & t < 0,\\[4pt]
1 - \frac{1}{2} e^{-t/\lambda}, & t \ge 0,
\end{cases}$
    \State $\mathcal{M}(f(X))_i \gets \xi (Z_i + \gamma_i)$
\EndFor
\State \Return $\mathcal{M}(f(X))$
\end{algorithmic}
\end{algorithm}

\begin{restatable}{theorem}{MainUpperBound}\label{thm:main_dithered_laplace}
Suppose \(f:\mathbb{X}^*\rightarrow\mathbb{R}^d\) has sensitivity
\(\Delta_1(f)\).
For every \(\varepsilon > 0\), \(\beta\in (0,1)\), and
\[
\alpha > \frac{2\Delta_1(f)\log(d/\beta)}{\varepsilon},
\]
there exists a random variable \(U\) and mechanisms
\(\mathcal{M}_U: \mathbb{X}^* \rightarrow\mathbb{R}^d\), defined by
post-processing the Laplace mechanism on \(f(X)\), that are \(\varepsilon\)-differentially private for
every fixed value of \(U\).
Moreover, \(\mathcal{M}_U(f(X))\) is an
\((\alpha,\beta)\)-accurate approximation of \(f(X)\), and its expected private
randomness complexity satisfies
\[
    \mathbb{E}_U\!\big[H(\M_U(f(X)))\big]
    =
    O\!\left(
    \frac{d\,\Delta_1(f)}{\alpha \varepsilon}
    \log\!\left(\frac{\alpha \varepsilon}{\Delta_1(f)}\right)
    \right).
\]
\end{restatable}

Compared with Corollary~\ref{cor:csv-general-alpha}, our bound removes the additive \(\log(d)\) term and improves the leading term by a factor of \(\log(d)\).
For error tolerance \(\alpha = \omega(d \Delta_1(f) \log(d) / \eps)\), the private randomness complexity becomes \(\mathbb{E}_U\!\big[H(\M_U(f(X)))\big] = o(1)\) independent of \(X\).
That is, the mechanism requires only a negligible amount of fresh, private randomness to run.
For error \(\alpha = \Theta(\Delta_1(f) \log(d/\beta)/\eps)\), matching the \(\ell_\infty\)-error of the standard Laplace mechanism, the private randomness complexity is \(\mathbb{E}_U\!\big[H(\M_U(f(X)))\big] = O(d \log\log(d) / \log d)\) which is \(o(d)\).
The mechanism's public randomness \(U\) can be sampled using \(O(\log(d))\) random bits with high probability, matching the lower bound~\eqref{eq:csv-lower} on total randomness.
We prove Theorem~\ref{thm:main_dithered_laplace}, together with a bound on public randomness complexity, in the following lemma.

\begin{restatable}{lemma}{mainTheoremDetailed}\label{lem:main-theorem-detailed}
Let $f:\mathbb{X}^*\to\mathbb{R}^d$ have $\ell_1$-sensitivity $\Delta_1(f)$.
Consider the dithered Laplace mechanism $\M$ defined in Algorithm~\ref{alg:dithered-laplace} with parameters $\xi>0$, $\lambda>0$, and public randomness \(U=(a,b)\).
Fix $\beta\in(0,1)$.
The mechanism satisfies:
\begin{enumerate}[label=(\alph*)]
    \item \textbf{Pure differential privacy:} For every choice of $\varepsilon>0$, if $\lambda = \Delta_1(f)/\varepsilon$, then $\M$ is $\varepsilon$-differentially private and $(\xi,\beta)$-accurate with
    \[
      \xi = 2\,\Delta_1(f) \log(d/\beta)/\varepsilon.
\]
    \item \textbf{Private randomness complexity:}
    The mechanism has expected private entropy
    \[
    \E_U\!\big[H(\M_U(f(X)))\big]\ =\ O\!\left(
    \frac{d\,\lambda}{\xi}\,
    \log\!\frac{\xi}{\lambda}
    \right).
\]
    \item \textbf{Public randomness complexity:}
    If relaxed to output a special value \(\bot\) with probability at most \(\beta\), the mechanism can be implemented using $O(\log(d/\beta))$ bits of public randomness.
\item \textbf{Distributional properties:} If relaxed to use an unbounded number of public random bits, the mechanism can be implemented such that it uses $O(\log(d/\beta))$ public random bits with probability \(1-\beta\) and so that
    \begin{itemize}
    \item it is unbiased, \(\E[\M(f(X))]=f(X)\), and 
    \item its coordinates are pairwise independent, \(\text{Cov}(\M(f(X))_i,\M(f(X))_j)=0\) for all \(i\ne j\).
\end{itemize}
  \end{enumerate}
\end{restatable}

\begin{proof}
We split the proof in five parts, covering privacy, error, private randomness complexity, public randomness complexity, and distributional properties, respectively.
\medskip

{\bf Privacy.}
Conceptually the mechanism ensures \(\eps\)-differential privacy by adding independent Laplace noise \(y \sim \text{Laplace}(\Delta_1(f)/\eps)^d\) to the value \(f(X)\).
The \(\ell_1\) sensitivity of \(f\) is \(\Delta_1(f)\), so Laplace noise with scale \(\lambda = \Delta_1(f) / \eps\) ensures the required privacy guarantee for \(f(X) + y\).
The subsequent rounding step maps the noisy result to the nearest point on an \(\xi\)-spaced grid shifted by the public offsets \(\xi\,\gamma_i\) with \(\gamma_i=(a i + b)\bmod 1\) and \((a,b)\sim\mathrm{Uniform}([0,1)^2)\).
The public randomness \((a,b)\) is independent of the data, so post-processing preserves privacy for every fixed choice of \((a,b)\).
\medskip

{\bf Error bound.}
As for the error incurred by the mechanism we note that
\[ \| \M(f(X)) - f(X) \|_\infty \leq \|y\|_\infty + \xi / 2, \]
since the error in each coordinate is bounded by the sum of the perturbation given by \(y\) and the rounding to the nearest grid point which changes the value by at most \(\xi/2\).
By a union bound over the coordinates, the probability that the noise exceeds \(\xi / 2\) is:
\[ \Pr[ \|y\|_\infty > \xi / 2 ] \leq d\, \exp\left( - \frac{\xi / 2}{\lambda} \right),\]
which is less than \(\beta\) when \(\xi \geq 2\lambda \log(d/\beta)\).
Thus,  with probability at least \(1-\beta\) the total error is bounded by \(\xi\).
Substituting \(\lambda = \Delta_1(f)/\epsilon\) completes the computation.

{\bf Private randomness complexity.}
Fix an arbitrary shift vector $\gamma\in\mathbb{R}^d$ and condition on $\gamma$ throughout this part; we only use its distribution at the end.
For a given input $z=f(X)$, define the (deterministic) baseline output
\[
\M_{\gamma,0}(z)\ :=\ \xi\cdot\Big(\Big\lfloor\frac{z}{\xi}-\gamma+\tfrac{1}{2}\Big\rfloor+\gamma\Big),
\]
obtained by setting the noise to zero.
Let $Y=(Y_1,\ldots,Y_d)\sim\mathrm{Laplace}(\lambda)^d$.
The private entropy for this input and shift is
\( H(\M_{\gamma,Y}(z)) \),
where $\gamma$ is the public randomness.
Since the coordinates of $Y$ are independent and rounding is coordinate-wise, $\M_{\gamma,Y}(z)$ is a deterministic function of independent integer ``step counters'' \(K_i\in\mathbb{Z}\) defined by
\(
 \M_{\gamma,Y}(z)_i\ =\ \M_{\gamma,0}(z)_i\ +\ K_i\,\xi
\).
Hence
\begin{equation}\label{eq:entropy-sum}
H(\M_{\gamma,Y}(z))\ \le H(K_1,\ldots,K_d\mid\gamma,z)
\ =\ \sum_{i=1}^d H(K_i\mid\gamma,z),
\end{equation}
where the last equality uses independence of \(Y\).
For $i\in[d]$ let $\tau_i=\tau_i(\gamma,z)\in[0,\xi/2]$ be the distance from $z_i$ to the nearest boundary of the randomly shifted grid on the $i$th coordinate.
That is, 
\[\tau_i:=\min_{k\in\mathbb{Z}}\{|z_i-\xi(k + \gamma_i + 1/2)|\}.\]
For \(p_{0,i}(\gamma,z) :=\ 1-e^{-\tau_i/\lambda} \) we have the following lower bound on the probability of $K_i=0$:
\[
\Pr[\,K_i=0\mid\gamma,z\,]\ \geq \Pr[\,|Y_i|<\tau_i\,]\ =\ p_{0,i}(\gamma,z).
\]
For $k\ge1$, crossing $k$ grid boundaries requires $|Y_i|\ge \tau_i+(k-1)\xi$, so the one-sided tail satisfies
\[
\Pr[\,|K_i|\ge k\mid\gamma,z\,]\ \le e^{-(\tau_i+(k-1)\xi)/\lambda}
\ =\ e^{-\tau_i/\lambda}\,q^{\,k-1},\text{ where } q:=e^{-\xi/\lambda}.
\]
Thus $|K_i|$ is stochastically dominated by a geometric distribution on $\{1,2,\ldots\}$ with ratio $q$, and the signed tail $\{K_i=\pm k: k\ge1\}$ by a two-sided geometric (discrete Laplace) with the same ratio.
Let $h_2$ be the binary entropy function, and let
\(
H_\mathrm{geo}(q)\ :=\ -\log_2(1-q)\ -\frac{q}{1-q}\,\log_2 q
\)
be the entropy of a geometric$\,(1-q)$ distribution on $\{1,2,\ldots\}$.
Encoding first the event $\{K_i=0\}$ and then the signed magnitude gives
\begin{equation}\label{eq:Hi-bound}
H(K_i\mid\gamma,z)\ \le h_2\big(p_{0,i}(\gamma,z)\big)\ +\ \big(1-p_{0,i}(\gamma,z)\big)\,\Big(\,1\ +\ H_\mathrm{geo}(q)\,\Big).
\end{equation}
Combining \eqref{eq:Hi-bound} over $i$ with \eqref{eq:entropy-sum} yields a per-input, per-shift entropy bound.
\medskip

\emph{Averaging over the shift.}  In our implementation the public shift is per-coordinate: \(\xi\,\gamma_i\) with \(\gamma_i=(a i + b)\bmod 1\) and \((a,b)\sim\mathrm{Uniform}([0,1)^2)\).
For each fixed input \(z\) and coordinate \(i\), the marginal \(\gamma_i\sim\mathrm{Uniform}([0,1))\) is independent of \(Y_i\), so the remainder \((z_i-\xi\gamma_i)/\xi \text{ mod } 1\) is uniform on \([0,1)\).
Consequently, the distance \(\tau_i(\gamma,z)\) from \(z_i\) to the nearest half-integer boundary is uniform on \([0,\xi/2]\).
Hence
\[
\overline{p}_0\ :=\ \EE_{\gamma}[\,p_{0,i}(\gamma,z)\,]\ =\ \EE_{T\sim\mathrm{Uniform}([0,\xi/2])}[\,1-e^{-T/\lambda}\,]
\ =\ 1\ -\ \frac{2\lambda}{\xi}\Big(1-e^{-\xi/(2\lambda)}\Big),
\]
independent of $i$ and $z$.
By concavity of $h_2$ and linearity of expectation,
\begin{equation}\label{eq:per-coordinate-avg}
\EE_{\gamma}\big[\,H(K_i\mid\gamma,z)\,\big]
\ \le h_2(\overline{p}_0)\ +\ (1-\overline{p}_0)\,\Big(\,1\ +\ H_\mathrm{geo}(q)\,\Big),
\qquad q=e^{-\xi/\lambda}.
\end{equation}
Finally,
\begin{equation}\label{eq:entropy-final}
\EE_U\big[H(\M_U(z))\big]
\ \le d\,\Big[\,h_2(\overline{p}_0)\ +\ (1-\overline{p}_0)\,\big(1+H_\mathrm{geo}(q)\big)\,\Big],
\end{equation}
with $\overline{p}_0=1-\tfrac{2\lambda}{\xi}(1-e^{-\xi/(2\lambda)})$ and $q=e^{-\xi/\lambda}$.
\medskip
Combining the above, we obtain an explicit asymptotic bound on the expected private entropy.
Under the accuracy condition $\xi \ge 2\lambda \log(d/\beta)$ established above (implying $\xi/\lambda \gg 1$ in the regimes of interest), the geometric ratio $q=e^{-\xi/\lambda}$ satisfies $q\ll1$, and
\(1-\overline{p}_0\ =\ \tfrac{2\lambda}{\xi}\big(1+O(e^{-\xi/(2\lambda)})\big)\).
Using $h_2(1-x)\le x\log_2(\tfrac{e}{x})$ and $H_\mathrm{geo}(q)=\Theta(q\log\tfrac{1}{q})$ for small $q$, inequality~\eqref{eq:entropy-final} yields
\[
\EE_U[H(\M_U(z))]\ \le C\,d\cdot\frac{\lambda}{\xi}\,\log_2\frac{\xi}{\lambda}
\]
for some absolute constant $C>0$.
This establishes the claimed bound in part~(b).
\medskip

{\bf Public randomness complexity (for part (c)).}
Fix $f(X)$ and the private noise $y$.
The output changes only when, for some coordinate $i\in[d]$, the function
\[
\gamma_i \mapsto \frac{f(X)_i+y_i}{\xi}-\gamma_i+\tfrac12
\ =\ \frac{f(X)_i+y_i}{\xi}-\big((a\cdot i + b)\bmod 1\big)+\tfrac12
\]
crosses an integer.
Equivalently, $(a,b)\in[0,1)^2$ lies in a set where $(a i + b)\bmod 1$ is within distance~$\eta$ of a (fixed) boundary point $r_i\in[0,1)$ determined by $f(X)_i$ and $y_i$.
For a precision parameter $\eta\in(0,1)$, let
\[
\mathcal{B}_\eta\ :=\ \bigcup_{i=1}^d \Big\{(a,b)\in[0,1)^2\ :\ \dist\big((a i + b)\bmod 1,\ r_i\big)\le \eta\Big\}.
\]
Each set inside the union is the preimage (under the map $(a,b)\mapsto (a i+b)\bmod 1$) of an interval of length $2\eta$ in $[0,1)$, so its total area in $[0,1)^2$ is at most $2\eta$.
By the union bound, the probability mass of
\(\mathcal{B}_\eta\) is at most \(2d\,\eta\).
Now choose $\eta:=\beta/(2d)$.
Then $\mu(\mathcal{B}_{\eta})\le \beta$.
If we quantize $a,b$ to $m$ bits each with
\[
    m=\left\lceil \log_2\frac{4d^2}{\beta}\right\rceil,
\]
then the rounded pair $(\tilde a,\tilde b)$ satisfies
\[
    \dist\big((\tilde a i+\tilde b)\bmod 1,(a i+b)\bmod 1\big)
    \le (i+1)2^{-m}
    \le \eta
\]
for every $i\in[d]$.
Thus, for every $(a,b)\notin\mathcal{B}_{\eta}$, the rounded pair $(\tilde a,\tilde b)$ determines the same output as $(a,b)$.
Therefore, an implementation that outputs $\bot$ whenever $(a,b)\in\mathcal{B}_{\eta}$ uses at most $2m=O(\log(d/\beta))$ public bits and triggers $\bot$ with probability at most $\beta$.
This proves part (c).
\medskip

{\bf Public randomness complexity and distributional properties (for part (d)).}
If we allow unbounded public randomness, we can reveal the bits of $(a,b)$ \emph{incrementally} until $(a,b)\notin\mathcal{B}_{\eta}$ is certified; with probability at least $1-\beta$ this occurs by $m=\lceil\log_2(4d^2/\beta)\rceil$ bits per parameter, so with probability $1-\beta$ the public randomness used is $O(\log(d/\beta))$ bits (and on the rare event we reveal more bits).
Since $(a,b)$ are uniform on $[0,1)^2$, the induced shifts $\gamma_i=(a i + b)\bmod 1$ satisfy:
  (i) for each $i$, $\gamma_i\sim\mathrm{Uniform}([0,1))$; and
  (ii) for $i\neq j$, the pair $(\gamma_i,\gamma_j)$ is uniform on $[0,1)^2$.
This \emph{pairwise independence} property is a variant of classical constructions of pairwise independent hash functions~\citep{carter1979random}, see Lemma~\ref{lem:affine-pairwise} in Appendix~\ref{app:affine} for a proof.
Consequently, with the round-to-nearest construction,
the per-coordinate quantization errors $E_i=\M(f(X))_i-(f(X)_i+y_i)$ are pairwise independent and uniform on $[-\xi/2,\xi/2]$, independent of $(f(X),y)$, and in particular have mean zero.
This proves part~(d).
\end{proof}

\section{Pairwise Independence of Affine Shifts}\label{app:affine}

We now show the pairwise independence of the shift values \(\gamma\) in the dithered rounding scheme.
The proof is analogous to Carter and Wegman's classical proof of pairwise independence of random linear functions over finite fields~\citep{carter1979random}; we include it for completeness.
\begin{lemma}\label{lem:affine-pairwise}
Let $a,b\sim\mathrm{Uniform}([0,1))$ be independent, and for each integer $t$ define
\(
\gamma_t := (a t + b)\bmod 1 \in [0,1).
\)
Then for every $t$ we have $\gamma_t\sim\mathrm{Uniform}([0,1))$, and for distinct integers $t\neq s$ the pair $(\gamma_t,\gamma_s)$ is uniform on $[0,1)^2$.
In particular, $\gamma_t$ and $\gamma_s$ are (pairwise) independent.
\end{lemma}
\begin{proof}
We work on the unit interval $[0,1)$ with addition taken modulo $1$.
Write $\{x\}$ for the fractional part of a real $x$.
Throughout, $a,b$ are independent and uniform on $[0,1)$.
\smallskip\noindent\emph{Uniform marginal.}
Fix an integer $t$.
For any measurable set $I\subseteq[0,1)$ and any fixed $a\in[0,1)$,
\[
\PP_b\!\big[\ \gamma_t\in I \bigm| a\ \big]\ =\ \PP_b\!\big[\,\{ta+b\}\in I\,\big]\ =\ \PP_b\!\big[\,b\in I-\{ta\}\,\big]\ =\ |I|,
\]
since $b\mapsto \{b+c\}$ is a shift on the circle and preserves the total length of sets in $[0,1)$.
Averaging over $a$ gives $\PP[\gamma_t\in I]=|I|$, i.e., $\gamma_t\sim{\rm Unif}[0,1)$.
\smallskip\noindent\emph{Pairwise independence.}
Let $t\neq s$ be integers, and let $I,J\subseteq[0,1)$ be measurable.
We show
\[
\PP\big[\ \gamma_t\in I,\ \gamma_s\in J\ \big]\ =\ |I|\,|J|.
\]
Condition on $a$.
For fixed $a$, using $b$ uniform and the shift argument as above,
\begin{align*}
\PP_b\!\big[\,\gamma_t\in I,\ \gamma_s\in J \bigm| a\,\big]
&= \PP_b\!\big[\fracpart(b+ta)\in I\quad\text{and}\quad\fracpart(b+sa)\in J\,\big] \\
&= \big|\, (I - \fracpart(ta)) \cap (J - \fracpart(sa)) \,\big|.
\end{align*}
Set $\delta:=\fracpart((t-s)a)\in[0,1)$ and note that $(I-\fracpart(ta))\cap (J-\fracpart(sa))=J\cap (I-\delta)$ up to a common shift, hence the length above equals $\phi(\delta):=|\,J\cap (I-\delta)\,|$.
Therefore
\[
\PP\big[\ \gamma_t\in I,\ \gamma_s\in J\ \big]\ =\ \EE_a\big[\,\phi(\{(t-s)a\})\,\big].
\]
We claim that $\{(t-s)a\}$ is uniform on $[0,1)$ when $a$ is uniform (for any nonzero integer $t-s$).
Indeed, partition $[0,1)$ into $|t-s|$ intervals of length $1/|t-s|$; multiplication by $|t-s|$ maps each interval linearly onto $[0,1)$, so for any measurable $B\subseteq[0,1)$ the preimage $\{a:\{(t-s)a\}\in B\}$ has size $|B|$.
Hence $\{(t-s)a\}\sim{\rm Unif}[0,1)$.
It remains to compute the average of $\phi(\delta)$ over uniform $\delta$.
By averaging over all possible shifts and using that translating a set within [0,1) does not change its total length,
\[
\int_0^1 \phi(\delta)\,d\delta\ =\ \int_0^1 \Big( \int_0^1 \mathbf{1}_J(x)\,\mathbf{1}_I(x+\delta)\,dx \Big)\,d\delta
\ =\ \int_0^1 \mathbf{1}_J(x)\,\Big( \int_0^1 \mathbf{1}_I(x+\delta)\,d\delta \Big)\,dx
\ =\ |J|\cdot |I|.
\]
Therefore $\EE_a[\phi(\{(t-s)a\})]=|I|\,|J|$, and the desired equality follows.
Since this holds for all measurable $I,J$, the pair $(\gamma_t,\gamma_s)$ is uniform on $[0,1)^2$, and in particular $\gamma_t$ and $\gamma_s$ are independent.
\end{proof}

\end{document}

%% file: math_commands.tex
\usepackage{amsmath,amsfonts,bm}

\def\eqref#1{equation~\ref{#1}}

\def\1{\bm{1}}

\def\eps{{\epsilon}}

\DeclareMathAlphabet{\mathsfit}{\encodingdefault}{\sfdefault}{m}{sl}
\SetMathAlphabet{\mathsfit}{bold}{\encodingdefault}{\sfdefault}{bx}{n}

\newcommand{\E}{\mathbb{E}}

